\documentclass[a4paper,twocolumn,11pt,accepted=2025-09-21]{quantumarticle}
\pdfoutput=1
\usepackage[utf8]{inputenc}
\usepackage[english]{babel}
\usepackage[T1]{fontenc}
\usepackage{amsmath}
\usepackage{hyperref}
\usepackage{soul}
\usepackage{dsfont}
\usepackage{diagbox}
\usepackage{amsthm}
\usepackage{amsfonts}

\usepackage{tikz}
\usepackage{lipsum}

\newtheorem{theorem}{Theorem}
\newtheorem{lemma}{Lemma}
\newtheorem{definition}{Definition}

\newcommand{\ket}[1]{\left| #1 \right>}
\newcommand{\bra}[1]{\left< #1 \right|}

\DeclareUnicodeCharacter{2009}{ } % fixing unicode character in references

\begin{document}

\title{A Mathematical Structure for Amplitude-Mixing Error-Transparent Gates for Binomial Codes} 

\author{Owen C. Wetherbee}
\email{ocw6@cornell.edu}
\affiliation{Department of Physics, Cornell University, Ithaca, NY, 14853, USA}

\author{Saswata Roy}
\affiliation{Department of Physics, Cornell University, Ithaca, NY, 14853, USA}

\author{Baptiste Royer}
\email{baptiste.royer@usherbrooke.ca}
\affiliation{Département de Physique and Institut Quantique, Université de Sherbrooke, Sherbrooke J1K 2R1, QC, Canada}

\author{Valla Fatemi}
\email{vf82@cornell.edu}
\affiliation{School of Applied and Engineering Physics, Cornell University, Ithaca, NY, 14853, USA}

\maketitle

\begin{abstract}
Bosonic encodings of quantum information offer hardware-efficient, noise-biased approaches to quantum error correction relative to qubit register encodings. 
Implementations have focused in particular on error correction of stored, idle quantum information, whereas quantum algorithms are likely to desire high duty cycles of active control.
Error-transparent operations are one way to preserve error rates during operations, but, to the best of our knowledge, only phase gates have so far been given an explicitly error-transparent formulation for binomial encodings.
Here, we introduce the concept of `parity nested' operations, and show how these operations can be designed to achieve continuous amplitude-mixing logical gates for binomial encodings that are fully error-transparent to the photon loss channel.
For a binomial encoding that protects against $l$ photon losses, the construction requires $\lfloor{l/2}\rfloor + 1$ orders of generalized squeezing in the parity nested operation to fully preserve this protection.
We further show that error-transparency to all the correctable photon jumps, but not the no-jump errors, can be achieved with just a single order of squeezing.
Finally, we comment on possible approaches to experimental realization of this concept. 
\end{abstract}

\section{Introduction}\label{sec:intro}
Logical quantum information schemes based on encoding logical qubits in the large Hilbert space of harmonic oscillators, so-called bosonic encodings, aim to provide hardware-efficient approaches to quantum error correction (QEC)~\cite{cai_bosonic_2021,terhal_towards_2020,liu_hybrid_2025}.
One key experimental platform is that of superconducting devices, in which harmonic oscillators demonstrate the lowest decoherence and a large noise bias to photon loss~\cite{reagor_quantum_2016,lei_high_2020,milul_superconducting_2023, ganjam_surpassing_2024}.
These approaches were the first among superconducting devices to demonstrate a gain from error correction~\cite{hu_quantum_2019,ni_beating_2023,sivak_real-time_2023,ofek_extending_2016}, owing in part to their low experimental overhead that simplifies various aspects of experiments. 

A special class of bosonic encodings are rotation-symmetric codes, such as the binomial encoding~\cite{michael_new_2016,grimsmo_quantum_2020}. 
These encodings support the logical and error subspaces on distinct parity (or generalized parity) manifolds of the oscillator Fock space.
Therefore, parity checks are sufficient for detecting photon loss errors in a correctable way, which is a good match for superconducting bosonic modes, which have single photon loss as the dominant source of errors.

So far, most demonstrations of error correction gain, on binomial codes and otherwise, are in the context of idling quantum information.
However, quantum algorithms will likely have a high duty cycle of logical operations, so the  design and implementation of the logical gate set are critical~\cite{hu_quantum_2019}.
Therefore, we must design operations that are fault tolerant, in which errors are not amplified if they occur during an operation (e.g., no increase in rate and no conversion to leakage errors).
A key class of fault tolerant gates are error-transparent (ET) gates~\cite{kapit_error-transparent_2018,ma_path-independent_2020}, which are designed to commute with the set of correctable errors within the codespace so that the errors pass through the algorithm with zero amplification.
So far, ET constructions for bosonic codes have focused on phase gates~\cite{ma_error-transparent_2020,reinhold_error-corrected_2020}, which were possible due to easier-to-implement operations that required tuning a number of parameters that scales linearly with the codespace dimension, $\mathcal{O}(d)$ (see Appendix~\ref{app:phasegates} for a discussion on this).

Here, we take a step forward by designing error-transparent gates for bosonic codes that transfer weight between Fock states, i.e., `amplitude-mixing operations'.
Specifically, we describe a method for constructing ET Hamiltonians which generate continuous logical rotations, like $\bar{X}(\theta)$, for binomial codes.
This problem is more challenging because the number of parameters scales quadratically with the codespace dimension, $\mathcal{O}(d^{2})$.
We introduce the concept of `parity nested' operations, illustrated in Fig.~\ref{fig:paritynested}, in which the Hamiltonian is block-diagonal in generalized photon number parity.
We then specifically design the couplings within each generalized parity manifold to produce amplitude-mixing gates for the binomial encodings that are fully error-transparent to the photon loss channel.
In particular, we define the ET distance of a gate to be the number of photon losses, and commensurate no-jump errors, that pass transparently through the operation.
For a binomial encoding that nominally protects against $l$ photon losses while idling, the construction requires $\lfloor{l/2}\rfloor +1$ orders of generalized squeezing to achieve an ET distance of $l$ (compare the solid and dotted couplings in Fig.~\ref{fig:paritynested}).
We further show that even with a single order of squeezing, it is possible to design an operation that is ET to all the correctable photon loss jumps, but not the no-jump errors.
Finally, in the concluding discussion, we comment on possible approaches and physical platforms for experimental realization of this concept.

\begin{figure*}
    \centering
    \includegraphics[width = \textwidth]{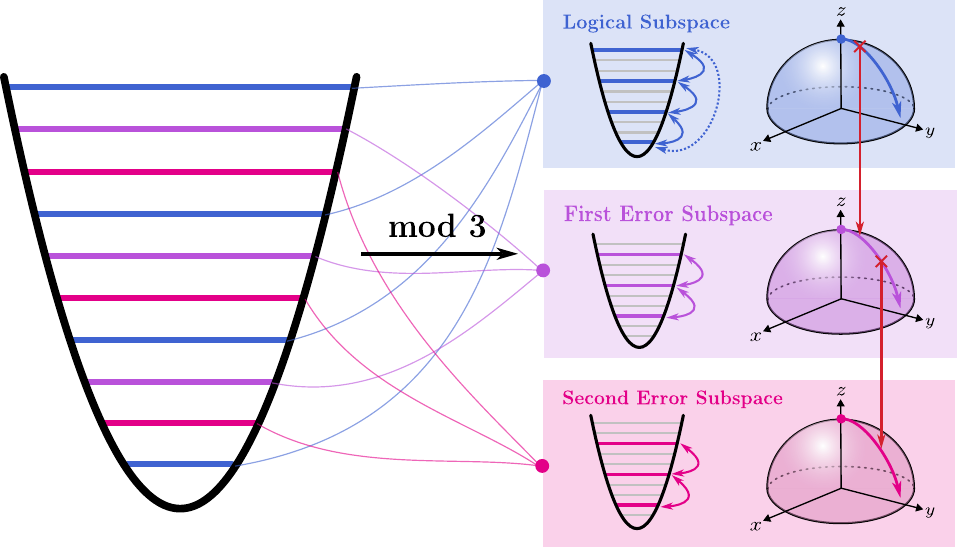}
    \caption{ 
        \textbf{Parity nesting schematic concept.}
        General procedure for constructing a `parity-nested' error-transparent $\bar{X}(\theta)$ gate for a binomial code, depicted here with order $N = 3$ and cutoff $K = 3$ (see Sec.~\ref{sec:rotSymmetric}).
        After breaking the Fock states into generalized mod $N$ parity manifolds, the couplings in each manifold are tuned to achieve identical $\bar{X}(\theta)$ evolution in the codespace and all error subspaces.
        The solid colored arrows depict the couplings required for error-transparency to the errors $\hat{a}$ and $\hat{a}^{2}$, and the dotted arrow is the extra coupling required to achieve a full ET distance of $2$ (i.e. also ET to the no-jump error $\hat{n}$).
        We note that couplings between Fock states separated by an odd number of states within the generalized parity manifolds do not need to be used.
        This is a consequence of the Fock-space structure of rotation-symmetric codes and in general results in fewer required orders of squeezing to achieve a given ET distance.
    }
    \label{fig:paritynested} 
\end{figure*}

\section{Background}\label{sec:background}
\subsection{Rotation-symmetric and binomial codes}\label{sec:rotSymmetric}
Rotation-symmetric codes~\cite{grimsmo_quantum_2020} are a broad class of bosonic encodings characterized by a discrete rotation symmetry in phase-space.
In particular, an order-$N$ bosonic rotation-symmetric code is defined as any code where the discrete rotation operator
\begin{equation}
    \begin{aligned}
        \hat{Z}_{N} = \exp\left(i\frac{\pi}{N}\hat{n}\right) ~,
    \end{aligned}
    \label{eqn:discreteRot}
\end{equation}
acts as a logical $\bar{Z}$ gate on the codespace, where $\hat{n}=\hat{a}^\dagger\hat{a}$ is the photon number operator.
This $N$-fold rotation symmetry forces the $\bar{Z}$ basis code words to have a Fock-basis structure of the form
\begin{equation}
    \begin{aligned}
        |0_{L}\rangle & = \sum_{k = 0}^{\infty}c_{2k}|2kN\rangle ~, \\
        |1_{L}\rangle & = \sum_{k = 0}^{\infty}c_{2k + 1}|(2k + 1)N\rangle ~.
    \end{aligned}
    \label{eqn:rotSymmetric}
\end{equation}
We collect these alternating Fock-grid coefficients $c_{2k}$ and $c_{2k + 1}$ into a single vector $\vec{c}$ that satisfies the normalization condition
\begin{equation}
    \begin{aligned}
        \sum_{k = 0}^{\infty}|c_{2k}|^{2} = \sum_{k = 0}^{\infty}|c_{2k + 1}|^{2} = 1 ~.
    \end{aligned}~
    \label{eqn:normalization}
\end{equation}
and defines a specific instance of an order-$N$ rotation-symmetric code. 
The dual-basis codewords $|\pm_{L}\rangle$ can also be written in terms of this vector of Fock-grid coefficients $\vec{c}$ as
\begin{equation}
    \begin{aligned}
        |\pm_{L}\rangle = \frac{1}{\sqrt{2}}\sum_{k = 0}^{\infty}(\pm1)^{k}c_{k}|kN\rangle ~.
    \end{aligned}
    \label{eqn:dualBasisRotSymmetric}
\end{equation}

Our error-transparent constructions will specifically apply to rotation-symmetric codes that have a hard energy cutoff $K$ such that the highest occupied Fock state is $|KN\rangle$ and the code is defined by a finite number of coefficients $\{c_{0}, \dots c_{K}\}$.
The principal example of such codes are the binomial codes, which are defined for each spacing $N \in \{1, 2, \dots\}$ and cutoff $K \in \{1, 2, \dots \}$ by the Fock-grid coefficients $c_{k} = \sqrt{\binom{K}{k}}$, where $\binom{K}{k}$ defines the binomial coefficient of $K$ choose $k$.
Binomial codes have the special property that they exactly satisfy the Knill-Laflamme conditions for the error set $\mathcal{A}_{l} = \{\hat{I}, \hat{a}, \dots, \hat{a}^{l}\}$, where $l = \min(N - 1, K - 1)$~\cite{michael_new_2016, grimsmo_quantum_2020}.
Thus, the Fock spacing parameter $N - 1$ specifies the number of detectable photon losses and the cutoff parameter $K - 1$ specifies the number of these photon losses that are also exactly correctable.

The error words associated with a given rotation-symmetric code are now defined.
In general, for any code with logical states $|\pm_{L}\rangle$ that corrects an (ordered) error set $\mathcal{E} = \{E_{1}, \dots, E_{n}\}$, the normalized and mutually orthogonal error words $|\pm_{E_{i}}\rangle$ can be iteratively defined for each $E_{i} \in \mathcal{E}$ as
\begin{equation}
    \begin{aligned}
        |\pm_{E_{i}}\rangle = \frac{\mathcal{P}_{i}E_{i}|\pm_{L}\rangle}{\sqrt{\langle\pm_{L}|E_{i}^{\dagger}\mathcal{P}_{i}E_{i}|\pm_{L}\rangle}} ~,
    \end{aligned}
  \label{eqn:errorWords}
\end{equation}
where $\mathcal{P}_{i} = \hat{I} - |\pm_{L}\rangle\langle\pm_{L}| - \sum_{j = 1}^{i - 1}|\pm_{E_{j}}\rangle\langle\pm_{E_{j}}|$ is the projector onto the subspace orthogonal to all the previous code and error words, which is what ensures these error subspaces are orthogonal.
In the case of rotation-symmetric codes, the error words with respect to the photon loss error set $\mathcal{E} = \mathcal{A}_{l} = \{\hat{I}, \hat{a}, \dots, \hat{a}^{l}\}$ are automatically orthogonal (i.e. the projector $\mathcal{P}_{i}$ can be omitted) and take on the form:
\begin{equation}
    \begin{aligned}
        |\pm_{\hat{a}^{m}}\rangle = \frac{\mathcal{N}_{m}}{\sqrt{2}}\sum_{k = 0}^{K}(\pm1)^{k}\varepsilon_{k}^{m}c_{k}|kN - m\rangle ~,
    \end{aligned}
    \label{eqn:errorWordsRotSym}
\end{equation}
where $\varepsilon_{k}^{m} = \sqrt{(kN - m + 1)\cdots(kN)}$ and $\mathcal{N}_{m}$ is a normalization constant that is the same for $|+_{\hat{a}^{m}}\rangle$ and $|-_{\hat{a}^{m}}\rangle$.
Notably, we can recognize these error words $|\pm_{\hat{a}^{m}}\rangle$ as themselves order-$N$ rotation-symmetric codes which have been `shifted' down $m$ levels in Fock space (i.e. obey Eq.~\eqref{eqn:discreteRot} with the added global phase $e^{-i\pi m / N}$) and have new coefficients $(c_{m})_{k} = \mathcal{N}_{m}\sqrt{(kN - m + 1)\cdots(kN)}c_{k}$.
We can thus think of the code and error subspaces of a rotation-symmetric code as a series of shifted rotation-symmetric codes which fill up the generalized parity manifolds of the Hilbert space -- a conceptualization which comes up throughout our derivation of error-transparent gates in Sec.~\ref{sec:parityNested}.

\subsection{Pure-loss bosonic channel}\label{sec:errorChannel}
The pure-loss bosonic channel is the main error channel that rotation-symmetric codes are designed to protect against. 
This loss channel can be defined in terms of the Kraus operators~\cite{chuang_bosonic_1997,ueda_probability-density-functional_1989,lee_superoperators_1994}
\begin{equation}
  \begin{aligned}
    E_{m} = \left(\frac{\gamma}{1 - \gamma}\right)^{m/2}\frac{\hat{a}^{m}}{\sqrt{m!}}(1 - \gamma)^{\hat{n}/2} \hspace{10pt} \forall m \in \mathbb{N}~,
  \end{aligned}
  \label{eqn:photonLoss}
\end{equation}
where $\gamma = 1 - e^{-\kappa t}$ is the probability of losing a photon after time $t$ when starting with one photon, given an excitation loss rate $\kappa$.
We are interested in codes which can correct these errors up to a specified order in $\gamma$, so it is useful to expand these errors in orders of $\sqrt{\gamma}$, which we do in Table~\ref{tab:errors}.

A quantum code corrects this set of errors up to $l$-th order in $\gamma$ if the following approximate Knill-Laflamme error-correction conditions are satisfied~\cite{grassl_quantum_2018}:
\begin{equation}
  \begin{aligned}
    \langle \mu_{L} | E_{i}^{\dagger}E_{j} | \nu_{L} \rangle = \delta_{\mu\nu}a_{ij} + \mathcal{O}(\gamma^{l + 1}) \hspace{10pt} \forall i, j \in \mathbb{N} ~,
  \end{aligned}
  \label{eqn:ECConditions}
\end{equation}
where $\mu, \nu$ $\in \{0, 1\}$ index the logical basis states, $\delta_{\mu\nu}$ is the Kronecker delta function, and $a_{ij}$ is a constant that does not depend on $\mu$ and $\nu$, but can depend on the error pair $E_{i}$ and $E_{j}$.
As discussed in~\cite{michael_new_2016,grassl_quantum_2018}, these conditions hold if all the leading terms of the errors $E_{m}$ up to order $\gamma^{l / 2}$ are mutually correctable.
In general, cross terms with the sub-leading terms involving factors of $\gamma^{k / 2}$ for $l < k < 2l$ also need to satisfy Eq.~\eqref{eqn:ECConditions}, but this always follows in the case of the pure-loss channel~\cite{michael_new_2016} and does not affect our error-transparent gate constructions.
 We can visualize these leading terms that need to be correctable as all those in the first $l + 1$ columns of Table~\ref{tab:errors}, up to and including column $l$.
In particular, we see that the $\gamma^{\frac{m}{2} + k}$ term of the $m$-th error $E_{m}$ involves $\hat{a}^{m}$ multiplied on the right by some polynomial in $\hat{n}$ of degree $k$.
Thus, for all the terms of order $\gamma^{l / 2}$ or less to be mutually correctable, we require the error set
\begin{equation}
    \begin{aligned}
        \mathcal{E}_{l} = \left\{\hat{a}^{m}\hat{n}^{k} \, \big| \, \frac{m}{2} + k \leq \frac{l}{2}\right\}
    \end{aligned}
    \label{eqn:ETSet}
\end{equation}
to exactly satisfy the Knill-Laflamme conditions.
We refer to the $\hat{n}^{k}$ factors of these errors as `no-jump' contributions, and the $\hat{a}^{m}$ errors as `pure jump' errors, so that the pure jump errors live on the main $m = l$ diagonal of Table~\ref{tab:errors}, and the no-jump contributions only appear for $l > m$ (i.e. above the main diagonal of Table~\ref{tab:errors}).
It turns out that, for any rotation-symmetric code, if just these `diagonal' pure photon jump errors $\mathcal{A}_{l} =\{\hat{I}, \hat{a}, \dots, \hat{a}^{l}\}$ are a correctable set, then the full error set $\mathcal{E}_{l}$ will also be correctable~\cite{michael_new_2016}.
Thus, in addition to protecting against the first $l$ photon jump errors $\mathcal{A}_{l}$, the binomial codes with $N, K > l$ also protect against the full pure-loss channel (including the $l > m$ no-jump contributions) to an accuracy of $\mathcal{O}(\gamma^{l + 1})$.

\begin{table*}
    \centering
    \renewcommand{\arraystretch}{1.5}
    \begin{tabular}{c|cccccc}
        \diagbox[linewidth = 0.65pt, width = 30pt, height = 25pt]{$m$}{$l$} & 0 & 1 & 2 & 3 & 4 & $\cdots$ \\
        \hline
        0 & $I$ & - & $-\frac{1}{2}\hat{n}$ & - & $\frac{1}{8}(\hat{n}^{2} - 2\hat{n})$ & $\cdots$ \\
        1 & - & $\hat{a}$ & - & $-\frac{1}{2}\hat{a}\hat{n}$ & - & $\cdots$ \\
        2 & - & - & $\frac{1}{\sqrt{2}}\hat{a}^{2}$ & - & $-\frac{1}{2\sqrt{2}}\hat{a}^{2}\hat{n}$ & $\cdots$ \\
        3 & - & - & - & $\frac{1}{\sqrt{6}}\hat{a}^{3}$ & - & $\cdots$ \\
        4 & - & - & - & - & $\frac{1}{2\sqrt{6}}\hat{a}^{4}$ & $\cdots$ \\
        $\vdots$ & $\vdots$ & $\vdots$ & $\vdots$ & $\vdots$ & $\vdots$ & $\ddots$ \\
    \end{tabular}
    \caption{\textbf{Expansions of error operators.} The $m$-th row gives the expansion of the $E_{m}$ Kraus operator of the pure-loss bosonic channel in orders of $l / 2$, where $l$ is the column index. For example, we can read off $E_{0} = I - \frac{1}{2}\hat{n}\gamma + \frac{1}{8}(\hat{n}^{2} - 2\hat{n})\gamma^{2} + \mathcal{O}(\gamma^{3})$ and $E_{1} = \hat{a}\sqrt{\gamma} - \frac{1}{2}\hat{a}\hat{n}\gamma^{3 / 2} + \mathcal{O}(\gamma^{5 / 2})$.}
    \label{tab:errors}
\end{table*}

\subsection{Error-transparent gates}\label{sec:ETgates}
We colloquially say a gate is fault tolerant if it does not amplify errors that occur before or during the gate.
Error-transparent gates are operations that achieve such fault tolerance by ensuring the corresponding unitary acts identically within all the logical and error subspaces.
In particular, no matter when an error occurs during the gate, the unitary evolution continues within the error subspace exactly as in the logical subspace, so that the result is the same as if the error had occurred at the end of the gate.
This can be formulated as a condition on the Hamiltonian generating the unitary evolution: the Hamiltonian must commute with the error within the logical subspace.
Specifically, we say that a unitary $U(t)$ generated by Hamiltonian $H(t)$ is error-transparent to a correctable error set $\mathcal{E}$ if~\cite{kapit_error-transparent_2018,vy_error-transparent_2013}
\begin{equation}
    \begin{aligned}
        \left[E, H(t)\right]|\mu_{L}\rangle = 0\hspace{10pt}\forall E \in \mathcal{E}, \mu \in \{0, 1\}, t ~.
    \end{aligned}
    \label{eqn:ETCondition}
\end{equation}
When $\mathcal{E} = \mathcal{E}_{l}$, defined in Eq.~\eqref{eqn:ETSet}, we say that the operation has an error-transparency distance of $l$, since the resulting unitary is $l$-th order error-transparent to the full pure-loss bosonic channel (i.e. the unitary is error-transparent to all the errors which must be correctable for $E_{m}$ to satisfy Eq.~\eqref{eqn:ECConditions}).
Although formulated separately for each error $E \in \mathcal{E}$, we note that these conditions ensure that even when multiple errors occur all at different times of the evolution, they will still all commute through the unitary in the sense described above as long as the resulting product of errors is correctable.\footnote{We also remark that there are ways to slightly relax this error-transparency condition, such as allowing the accumulation of different global phases~\cite{ma_error-transparent_2020} or requiring the errors to simply remain correctable after the gate (instead of remaining the same)~\cite{tsunoda_error-detectable_2023}.
However, we have found that the simpler and stricter conditions given by Eq.~\eqref{eqn:ETCondition} are achieved by our construction.}

Logical phase gates $\bar{Z}(\theta)$ for binomial codes have already been made error-transparent to one photon loss $\mathcal{E}_{1} = \mathcal{A}_{1} = \{\hat{I}, \hat{a}\}$.
This can be achieved by using Stark shifts~\cite{ma_error-transparent_2020}, or SNAP gates~\cite{reinhold_error-corrected_2020}\footnote{Although this work focused on error-transparency to ancilla errors, such SNAP gates could in principle be used to achieve error-transparency to one bosonic photon loss.}, to design diagonal Hamiltonians of the form $H = \sum_{n = 0}^{n_{\text{trc}}}\varphi_{n}|n\rangle\langle n|$.
Leveraging the Fock-space structure of rotation-symmetric codes, these operations directly add different phases to each code and error words to accomplish $\bar{Z}(\theta)$ gates in both the logical and error subspaces, yielding error-transparent phase gates (see Appendix~\ref{app:phasegates}).
This cannot, however, be easily extended to amplitude-mixing gates since the requisite non-diagonal matrix elements of $H$ do not interact simply under commutation with the $\sqrt{n}$ factor of $\hat{a}$ or the $\hat{n}^{k}$ no-jump contributions of the errors in $\mathcal{E}_{l}$.

\section{Amplitude-Mixing Gates Using Parity Nested Operations}\label{sec:parityNested}
We will design amplitude-mixing error-transparent gates for the binomial codes by leveraging the periodic spacing in Fock space of their codewords.
In particular, recall that (see Sec.~\ref{sec:background}) the logical and error subspaces of rotation-symmetric codes can be segregated into different $\left|-m \bmod N\right\rangle$ generalized parity manifolds depending on the power of $\hat{a}$ in the error $\hat{a}^{m}\hat{n}^{k}$: the codespace has support only on the $|0 \bmod N\rangle$ Fock states, the $\hat{a}$ error subspace has support only on the $\left|-1 \bmod N\right\rangle$ Fock states, etc.
Here, $\left|-m \bmod N\right\rangle$ denotes the set of all Fock states $|n\rangle$ whose number $n$ is equivalent to $-m$ modulo $N$.
Operations which preserve generalized mod $N$ parity at each time step will then always keep the logical and error subspaces within their corresponding parity manifolds and thus prevent them from mixing.
At the very least this means that detectable photon losses will remain detectable under parity-preserving operations.

We will take this one step further by designing parity-preserving operations which specifically act identically on each logical and error subspace of the binomial code, yielding gates which are error-transparent up to some desired distance $l < N$.
We call such parity-preserving error-transparent gates `parity nested gates,' and they can theoretically be implemented by tuning the matrix elements of generalized higher-order squeezing operations (Sec.~\ref{sec:paritysqueezing}).
In the following sections, we give two general constructions of parity nested $\bar{X}(\theta)$ gates for binomial codes: one achieves error-transparency to the full error set $\mathcal{E}_{l}$ using $\left\lfloor l / 2 \right\rfloor + 1$ orders of squeezing (Sec.~\ref{sec:squeezingScaling}), and the other achieves error-transparency to just the pure-jump errors $\mathcal{A}_{l} = \{\hat{I}, \hat{a}, \dots, \hat{a}^{l}\}$ using only a single order of squeezing (Sec.~\ref{sec:singleSqueezing}).

\subsection{Parity-preserving operations and generalized squeezing}\label{sec:paritysqueezing}
Any unitary which preserves generalized mod $N$ parity at each time step must have a corresponding Hamiltonian $H$ which can be decomposed into independent Hamiltonians $H_{m}$ acting on each respective $\left|-m \bmod N\right\rangle$ parity manifold:
\begin{equation}
    \begin{aligned}
        H = \sum_{m = 0}^{N - 1}H_{m} ~,
    \end{aligned}
    \label{eqn:parityBlocked}
\end{equation}
where
\begin{equation}
    \begin{aligned}
        H_{m} = \sum_{k, k' = 0}^{K}(H_{m})_{k', k}|k'N - m\rangle\langle kN - m| ~.
    \end{aligned}
    \label{eqn:parityBlockedContinued}
\end{equation}
When $m > 0$ the summations start at $k, k' = 1$ so that $H$ does not involve any negative Fock states.
Defined in this way, $H_{0}$ is a $(K + 1) \times (K + 1)$ matrix that determines the unitary evolution in the $|0 \bmod N\rangle$ logical parity manifold, and the remaining $H_{m}$ are $K \times K$ matrices (since the $|0\rangle$ Fock state has been annihilated) that determine the evolution in the corresponding $\left|-m \bmod N\right\rangle$ $m$-th error parity manifolds.
Intuitively, we can think of $H$ as the `parity-blocked' matrix formed by stacking these separate parity manifold Hamitonians $H_{m}$ block-diagonally:
\begin{equation}
    \begin{aligned}
        H & = M^{-1} \begin{bmatrix} H_{0} & 0 & 0 & 0 \\ 0 & H_{1} & 0 & 0 \\ 0 & 0 & \ddots & 0 \\ 0 & 0 & 0 & H_{N - 1} \end{bmatrix} M \\
        & = M^{-1}H'M~,
    \end{aligned}
    \label{eqn:parityBlocked2}
\end{equation}
where $M = \sum_{n = 0}^{KN}\left|\left\lceil n / N \right\rceil + \left(-n \bmod N\right)K\right\rangle\langle n|$\footnote{This usage of mod refers to the modulo operation (i.e. the remainder of $-n$ divided by $N$).} is the permutation matrix that rearranges the Fock states so that those which are congruent mod $N$ are all adjacent to each other. We note that we can only write a Hamiltonian in this block diagonal form if it is parity nested and thus does not contain matrix elements which mix two different generalized parity manifolds.

Thinking about how to actually implement these operations, the simplest type of nontrivial unitary that preserves even/odd (i.e. mod $2$) parity at each time step is the squeeze operator $\hat{S}(z) = \exp\left(\frac{1}{2}\left(z^{*}\hat{a}^{2} - z(\hat{a}^{\dagger})^{2}\right)\right)$, whose Hamiltonian only has matrix elements on the two second off-diagonals.
This can be extended to generalized mod $N$ parity-preserving operations by going to `higher order' squeezing: $\hat{S}_{N}(z) = \exp\left(\frac{1}{2}\left(z^{*}\hat{a}^{N} - z(\hat{a}^{\dagger})^{N}\right)\right)$.
Writing the Hamiltonian $H = z^{*}\hat{a}^{N} + z(\hat{a}^{\dagger})^{N}$ of this unitary in the parity-blocked form of Eq.~\eqref{eqn:parityBlocked2} results in a block-diagonal matrix $H'$ containing $N$ sub-matrices $H_{0}, H_{1}, \dots, H_{N - 1}$ (since $H$ preserves mod $N$ parity), each only having nearest-neighbor couplings (since $H$ only has matrix elements on the two $N$-th off-diagonals).
We emphasize that these `bare' bosonic squeezing operators are not the ones that accomplish the relevant gates for binomial codes, but they provide a useful conceptualization framework for parity nesting.
In particular, in a similar spirit to~\cite{rojkov_stabilization_2024_2}, the matrix elements of the $\hat{a}$ and $\hat{a}^{\dagger}$ operations need to be tuned to design more general Hamiltonians of the form
\begin{equation}
    H = f(\hat{n})^{*}\hat{a}^{N} + f(\hat{n})(\hat{a}^{\dag})^{N} ~,
\end{equation}
where $f(\hat{n})$ is some diagonal matrix in Fock space which effectively changes the $\sqrt{n}$ matrix elements of the $\hat{a}$ operator.

Unfortunately, when implementing actual gates, the parity manifold Hamiltonians $H_{m}$ may require nonzero matrix elements on not only these two first off-diagonals, corresponding to the $N$-th off-diagonals of $H$, but also on the other off-diagonals as well, corresponding to the $2N$-th, $3N$-th, etc. off-diagonals of $H$.
Physically realizing such Hamiltonians would require the simultaneous implementation of several different layers of higher order squeezing -- one for each $kN$-th off-diagonal of $H$ that contains nonzero matrix elements.
With this in mind, we classify the complexity of a parity nested gate by the number of nonzero off-diagonals in its Hamiltonian, and therefore the number of theoretical orders of squeezing that would be needed to implement its Hamiltonian, assuming the nonzero matrix elements could be arbitrarily tuned.

\subsection{Parity nested \texorpdfstring{$\bar{X}(\theta)$}{barX} gates for the binomial codes}\label{sec:derivingGates}
To serve as a starting point for our two main constructions, we now derive parity nested $\bar{X}(\theta)$ gates for the cutoff-$K$ order-$N$ binomial codes (with $N \geq K$) that have an error-transparency (ET) distance $l = K - 1$ commensurate with the order $\mathcal{O}(\gamma^{K})$ of protection achieved while idling.
In particular, for a Hamiltonian $H$ to generate continuous $\bar{X}(\theta)$ gates, it must have $|+_{L}\rangle$ in its $+1$ eigenspace and $|-_{L}\rangle$ in its $-1$ eigenspace.
For $H$ to also be error-transparent to $\mathcal{E}_{K - 1}$ (and thus have an ET distance of $K - 1$), it must act in this same way on all the error subspaces and must therefore likewise have $E|+_{L}\rangle$ in its $+1$ eigenspace and $E|-_{L}\rangle$ in its $-1$ eigenspace, for every error $E \in \mathcal{E}_{K - 1}$.
We can directly derive an $H$ with this eigenstructure by using the orthogonalized error words defined in Eq.~\eqref{eqn:errorWords}:
\begin{equation}
    \begin{aligned}
        H = \sum_{E \in \mathcal{E}_{K - 1}}\left(|+_{E}\rangle\langle+_{E}| - |-_{E}\rangle\langle-_{E}|\right) ~,
    \end{aligned}
  \label{eqn:HETgeneral}
\end{equation}
which generates a continuous $\bar{X}(\theta)$ gate in the logical subspace (since $I \in \mathcal{E}_{K - 1}$) and all the orthogonalized error subspaces. 
Thus, the Hamiltonian $H$ in Eq.~\eqref{eqn:HETgeneral} automatically satisfies the ET conditions in Eq.~\eqref{eqn:ETCondition}: $[H, E]|\pm_{L}\rangle = H\left(E|\pm_{L}\rangle\right) - E\left(H|\pm_{L}\rangle\right) = \pm E|\pm_{L}\rangle - E\left(\pm|\pm_{L}\rangle\right) = 0$ for all $E \in \mathcal{E}_{K - 1}$.
Since each of the errors words $|\pm_{\hat{a}^{m}\hat{n}^{k}}\rangle$ lives in a single parity manifold, dictated by the power $m$ of the $\hat{a}$ operator in the error, this $H$ will also always be parity nested.

To understand how many orders of squeezing are required for this construction, we can analyze each parity manifold independently.
Starting with the logical $|0 \bmod N\rangle$ parity manifold, only the error words arising from the no-jump errors $\hat{n}, \hat{n}^{2}, \dots, \hat{n}^{\left\lfloor \frac{K - 1}{2} \right\rfloor}$ live in this parity manifold (see Eq.~\eqref{eqn:ETSet} with $l = K - 1$ or the first, $m = 0$, row of Table~\ref{tab:errors}).
So, we have:
\begin{equation}
  H_{0} = \sum_{k = 0}^{\left\lfloor \frac{K - 1}{2} \right\rfloor}\left(|+_{\hat{n}^{k}}\rangle\langle+_{\hat{n}^{k}}| - |-_{\hat{n}^{k}}\rangle\langle-_{\hat{n}^{k}}|\right) ~.
  \label{eqn:H0ET}
\end{equation}
Since the $|0_{L}\rangle$ and $|1_{L}\rangle$ codewords of rotation-symmetric codes, and their corresponding $|0_{\hat{n}^{k}}\rangle$ and $|1_{\hat{n}^{k}}\rangle$ error words, have support on alternating Fock states (see Eq.~\eqref{eqn:rotSymmetric}), each $|+_{\hat{n}^{k}}\rangle\langle+_{\hat{n}^{k}}| - |-_{\hat{n}^{k}}\rangle\langle-_{\hat{n}^{k}}| = |0_{\hat{n}^{k}}\rangle\langle 1_{\hat{n}^{k}}| + |1_{\hat{n}^{k}}\rangle\langle 0_{\hat{n}^{k}}|$ pair of terms in Eq.~\eqref{eqn:H0ET} will only contribute nonzero matrix elements to the \textit{odd} off-diagonals of $H_{0}$.
For cutoff-$K$ binomial codes, $H_{0}$ is a $(K + 1) \times (K + 1)$ matrix, meaning it will have $\left\lfloor \frac{K + 1}{2} \right\rfloor$ of these odd off-diagonals (above the main diagonal).

Now considering the other parity manifold Hamiltonians $H_{m}$ arising from the jump errors $\hat{a}^{m}\hat{n}^{k}$, since the $|\pm_{\hat{a}^{m}}\rangle$ error words are just shifted rotation-symmetric codes, prepending factors of $\hat{a}$ to all the code and error words in Eq.~\eqref{eqn:H0ET} (and removing factors of $\hat{n}$ accordingly) results in essentially the same story for each error parity manifold Hamiltonian $H_{m}$.
This means the entire matrix $H$ will generically have nonzero matrix elements on all its $\left\lfloor \frac{K + 1}{2} \right\rfloor$ odd off-diagonals (above the main diagonal), and thus will require $\left\lfloor \frac{K + 1}{2} \right\rfloor$ squeezing orders to implement, scaling with the size $K$ of the binomial code.

\subsection{Number of squeezing orders scales with error-transparency distance}\label{sec:squeezingScaling}
Any Hamiltonian that generates continuous logical $\bar{X}(\theta)$ gates and is error-transparent (ET) to $\mathcal{E}_{l} = \mathcal{E}_{K - 1}$ must contain the eigenstructure in Eq.~\eqref{eqn:HETgeneral}.
However, we can potentially reduce the number of required squeezing orders by adding extra terms to this eigendecomposition.
Specifically, we can construct Hamiltonians of the form
\begin{equation}
    H = \sum_{E \in \mathcal{E}_{l}}\left(|+_{E}\rangle\langle +_{E}| - |-_{E}\rangle\langle -_{E}|\right) + \sum_{i = 1}^{D}\lambda_{i}|\psi_{i}\rangle\langle\psi_{i}| ~,
    \label{eqn:addedTerms}
\end{equation}
where $\lambda_{i} \in \mathbb{R}$, the $|\psi_{i}\rangle$ are orthogonal to each other and all the code and error words, and $D = (KN + 1) - 2|\mathcal{E}_{l}|$ is the dimension of the subspace orthogonal to all the code and error subspaces.
By picking these extra $\lambda_{i}$ eigenvalues and $|\psi_{i}\rangle$ eigenvectors judiciously, we can potentially cancel out all the matrix elements on certain off-diagonals of $H$, thus reducing the required number of squeezing orders.

Unfortunately, when the ET distance $l = K - 1$ matches the order of protection achieved while idling, the eigenstructure of $H_{0}$ (Eq.~\eqref{eqn:H0ET}) contains $2\left\lfloor \frac{K - 1}{2}\right\rfloor + 2 = K + 1$ terms (when $K$ is odd), meaning $H_{0}$ is already full-rank, so no additional terms of the form $\lambda_{i}|\psi_{i}\rangle\langle\psi_{i}|$ can be added to the $H_{0}$ sector of $H$ (the one extra term that can be added when $K$ is even is not sufficient to cancel out any matrix elements).
In other words, the two-dimensional error subspaces arising from $\hat{n}, \hat{n}^{2}, \dots, \hat{n}^{\left\lfloor \frac{K - 1}{2} \right\rfloor}$, along with the codespace, fill up all $K + 1$ dimensions of the logical parity manifold, leaving no extra subspaces on which $H_{0}$ can act arbitrarily to cancel out matrix elements.
In this sense, we say that the error set $\mathcal{E}_{K - 1}$ \emph{saturates} the logical parity manifold of a cutoff-$K$ binomial code, fixing $H_{0}$ and thus the $\left\lfloor \frac{K + 1}{2} \right\rfloor$ required number of squeezing orders.

However, if we instead only seek operations with a smaller ET distance $l < K - 1$, then $\mathcal{E}_{l}$ will no longer saturate the logical parity manifold.
This gives us the freedom to add $K + 1 - 2(\lfloor l / 2 \rfloor + 1) = K - l - 1$ arbitrary terms to the eigendecomposition in Eq.~\eqref{eqn:H0ET} (or $K - l$ terms in the odd $l$ case), which theoretically involves enough free variables to cancel out matrix elements so that $H$ only requires $\left\lfloor l / 2 \right\rfloor + 1$ orders of squeezing.
Although directly determining the necessary $\lambda_{i}$ and $|\psi_{i}\rangle$ requires solving a quadratic system of equations that is not necessarily consistent, we find that this can indeed be done through a different construction:

\begin{theorem}
  For any rotation-symmetric code with a finite number of nonzero coefficients $\vec{c} = [c_{0}, \dots, c_{K}]$ which exactly corrects the error set $\mathcal{E}_{l}=\left\{\hat{a}^{m}\hat{n}^{k} \, \big| \, \frac{m}{2} + k \leq \frac{l}{2}\right\}$, we can always construct a parity nested Hamiltonian $H$ generating continuous logical $\bar{X}(\theta)$ gates which is error-transparent to the error set $\mathcal{E}_{l}$ and only requires $\left\lfloor l / 2 \right\rfloor + 1$ orders of squeezing  to implement.
  \label{thm:SqueezingScaling}
\end{theorem}
We leave the full details of the proof for Appendix~\ref{app:thm1}, but the general idea of this construction is to first break down the code into a series of smaller shifted binomial codes of size $K' = l + 1$, which can be thought of as the `minimal correcting units' of the error set $\mathcal{E}_{l}$.
We can then separately apply the error-transparent construction from Eq.~\eqref{eqn:HETgeneral} to each of these smaller $K' = l + 1$ binomial codes, yielding decoupled Hamiltonians which each only utilize the $\left\lfloor \frac{K' + 1}{2} \right\rfloor = \left\lfloor \frac{l}{2} \right\rfloor + 1$ odd off-diagonals in their immediate vicinity in Fock space.
We then recombine these decoupled Hamiltonians to arrive at the desired $l$-distance ET Hamiltonian $H$, which still only requires $\left\lfloor l / 2 \right\rfloor + 1$ squeezing drives to implement, just like its constituent parts.
Thus, the required number of squeezing orders scales with the desired ET distance $l$, instead of the size $K$ of the code. 

We remark that, while the Hamiltonian constructed in Theorem.~\ref{thm:SqueezingScaling} is certainly not the unique one ET to $\mathcal{E}_{l}$, there is a meaningful sense in which it is the minimal one with respect to squeezing orders.
In particular, as we prove in Appendix.~\ref{sec:minSqueezingOrders}, given any rotation-symmetric code which corrects $\mathcal{E}_{l}$, any Hamiltonian which is error-transparent to $\mathcal{E}_{l}$, continuously maps the logical subspace to itself, and yields a continuous amplitude-mixing gate must require at least $\lfloor l / 2 \rfloor + 1$ orders of squeezing, specifically involving matrix elements on the odd off-diagonals in the $H_{m}$ parity blocks.
The lowest $\lfloor l / 2 \rfloor + 1$ such orders of squeezing are the only ones utilized by the Hamiltonians in the Theorem.~\ref{thm:SqueezingScaling} construction, which is why it is in some sense minimal with respect to squeezing orders.
More generally, this result suggests that these squeezing orders are likely required to achieve exactly error-transparent amplitude-mixing gates for rotation-symmetric codes.
That said, logical unitaries which temporarily leave the codespace, for which the concept of error-transparency becomes much less straightforward, may be able to circumvent this constraint.
For the specific statement of this required squeezing theorem, as well as its proof and further discussion, see Appendix.~\ref{sec:minSqueezingOrders}.

\subsection{Single order of squeezing}\label{sec:singleSqueezing}
Simultaneous implementation of more than one order of squeezing may be challenging.
This motivates the question: what type of error-transparency (ET) performance can be achieved using only a single order of squeezing?
From Theorem~\ref{thm:SqueezingScaling}, we know that with a single order of squeezing we can achieve first-order, but not second-order, error-transparency to the pure-loss channel.
We now improve this result by showing that, using only a single order of squeezing, we can implement Hamiltonians that are more generally error-transparent to the first $l$ correctable photon jumps:
\begin{theorem}
  For any rotation-symmetric code with a finite number of nonzero coefficients $\vec{c} = [c_{0}, \dots, c_{K}]$ which exactly corrects the error set $\mathcal{A}_{l} = \{\hat{I}, \hat{a}, \dots, \hat{a}^{l}\}$, consider the parity nested Hamiltonian $H$ formed from nearest-neighbor parity manifold Hamiltonians $H_{m}$ given by
  \begin{equation}
      \begin{aligned}
          (H_{m})_{k, k + 1} & = (H_{m})_{k + 1, k}^{*} \\
          & = \frac{1}{(c_{m})_{k}^{*}(c_{m})_{k + 1}}\sum_{j = 0}^{K}(-1)^{j + k}|(c_{m})_{j}|^{2} ~,
      \end{aligned}
      \label{eqn:thm2}
  \end{equation}
  where the error coefficients $(c_{m})_{k} = \mathcal{N}_{m}\sqrt{(kN - m + 1)\cdots(kN)}c_{k}$ are defined below Eq.~\eqref{eqn:errorWordsRotSym}. The resulting $H$ then generates continuous logical $\bar{X}(\theta)$ gates, is error-transparent to the error set $\mathcal{A}_{l}$, and only requires a single order of squeezing to implement.
  \label{thm:SingleSqueezing}
\end{theorem}
We now give a sketch of the proof, leaving a full derivation of Eq.~\eqref{eqn:thm2} to Appendix~\ref{app:thm2}.
Using Theorem~\ref{thm:SqueezingScaling} with $l = 0$, we can first derive a nearest-neighbor logical parity manifold $H_{0}$ that generates continuous logical $\bar{X}(\theta)$ gates in the codespace, yielding the $m = 0$ case of Eq.~\eqref{eqn:thm2}.
Then, since the error words with respect to pure-jump $\hat{a}^{m}$ errors are shifted rotation-symmetric codes (see Sec.~\ref{sec:rotSymmetric}), we can repeat this construction for each error parity manifold Hamiltonian $H_{m}$, yielding a full Hamiltonian $H$ which generates $\bar{X}(\theta)$ gates in both the codespace and all the $\hat{a}^{m}$ error spaces and only requires a single order of squeezing.
The resulting Hamiltonian (Eq.~\eqref{eqn:thm2}) indeed satisfies the ET conditions in Eq.~\eqref{eqn:ETCondition}.

This theorem tells us that the no-jump errors are what force the use of all the additional $\left\lfloor l / 2 \right\rfloor$ squeezing orders, beyond the first, required for full pure-loss error-transparency.
Although this theorem does not yield single-order squeezing Hamiltonians with ET distances greater than $1$ (since we lack error-transparency to the second-order no-jump error $\hat{n} \in \mathcal{E}_{2}$), it does allow us to suppress the higher-order jump errors (e.g. $\hat{a}^{2} \in \mathcal{E}_{2}$) in addition to $\hat{a} \in \mathcal{E}_{1}$. Since $\mathcal A_{l} \supset \mathcal E_{1}$, this indeed yields a substantial reduction in errors as we demonstrate in Sec.~\ref{sec:performance}.
A similar procedure can be used to also get error improvements over the construction in Theorem~\ref{thm:SqueezingScaling} when we are limited to a larger number of squeezing orders (more than one), but we leave the corresponding derivation to a future work.

\section{Examples and Performance}\label{sec:performance}
As an example and to test the performance of our scheme, we now apply the various parity nested constructions presented in Sec.~\ref{sec:parityNested} to the $N = 3$, $K = 3$ binomial code.
This binomial code has codewords
\begin{equation}
    \begin{aligned}
        |0_{L}\rangle & = \frac{|0\rangle + \sqrt{3}|6\rangle}{2} \\
        |1_{L}\rangle & = \frac{\sqrt{3}|3\rangle + |9\rangle}{2}
    \end{aligned}
\end{equation}
and exactly corrects the error set $\mathcal{E}_{2} = \{\hat{I}, \hat{n}, \hat{a}, \hat{a}^{2}\}$.

\emph{Construction 1.}
From Sec.~\ref{sec:derivingGates}, the parity nested $\bar{X}$ gate for this code that is error-transparent (ET) to the full error set $\mathcal{E}_{2}$ requires $\left\lfloor \frac{K + 1}{2} \right\rfloor = 2$ squeezing orders to implement.
In particular, using Eq.~\eqref{eqn:H0ET} we can write the logical parity manifold Hamiltonian as
\begin{equation}
    \begin{aligned}
        H_{0} & = |0_{L}\rangle\langle 1_{L}| + |1_{L}\rangle\langle 0_{L}| + |0_{\hat{n}}\rangle\langle 1_{\hat{n}}| + |1_{\hat{n}}\rangle\langle 0_{\hat{n}}| \\
        & = \frac{1}{2}\begin{bmatrix} 0 & \sqrt{3} & 0 & -1 \\ \sqrt{3} & 0 & 1 & 0 \\ 0 & 1 & 0 & \sqrt{3} \\ -1 & 0 & \sqrt{3} & 0 \end{bmatrix} ~.
    \end{aligned}
    \label{eqn:exampleH0}
\end{equation}
Likewise for the first and second error parity manifolds, we have:
\begin{equation}
    \begin{aligned}
        H_{1} & = |0_{\hat{a}}\rangle\langle 1_{\hat{a}}| + |1_{\hat{a}}\rangle\langle 0_{\hat{a}}| \hspace{30pt} \\
        & = \frac{1}{\sqrt{2}}\begin{bmatrix} 0 & 1 & 0 \\ 1 & 0 & 1 \\ 0 & 1 & 0 \end{bmatrix} ~, \\ \\
        H_{2} & = |0_{\hat{a}^{2}}\rangle\langle 1_{\hat{a}^{2}}| + |1_{\hat{a}^{2}}\rangle\langle 0_{\hat{a}^{2}}| \\
        & = \frac{1}{\sqrt{5}}\begin{bmatrix} 0 & 1 & 0 \\ 1 & 0 & 2 \\ 0 & 2 & 0 \end{bmatrix} ~.
    \end{aligned}
    \label{eqn:exampleH12}
\end{equation}
The full parity-blocked Hamiltonian $H$ (see Eq.~\eqref{eqn:parityBlocked2}) will then be given by
\begin{equation}
    \begin{aligned}
        & H = M^{-1}\begin{bmatrix} \color{blue}H_{0}\color{black} & 0 & 0 \\ 0 & \color{violet}H_{1}\color{black} & 0 \\ 0 & 0 & \color{red}H_{2}\color{black} \end{bmatrix}M = \\
        & \begin{bmatrix} 0 & 0 & 0 & \color{blue}\frac{\sqrt{3}}{2}\color{black} & 0 & 0 & 0 & 0 & 0 & \color{blue}-\frac{1}{2}\color{black} \\ 0 & 0 & 0 & 0 & \color{red}\frac{1}{\sqrt{5}}\color{black} & 0 & 0 & 0 & 0 & 0 \\ 0 & 0 & 0 & 0 & 0 & \color{violet}\frac{1}{\sqrt{2}}\color{black} & 0 & 0 & 0 & 0 \\ \color{blue}\frac{\sqrt{3}}{2}\color{black} & 0 & 0 & 0 & 0 & 0 & \color{blue}\frac{1}{2}\color{black} & 0 & 0 & 0 \\ 0 & \color{red}\frac{1}{\sqrt{5}}\color{black} & 0 & 0 & 0 & 0 & 0 & \color{red}\frac{2}{\sqrt{5}}\color{black} & 0 & 0 \\ 0 & 0 & \color{violet}\frac{1}{\sqrt{2}}\color{black} & 0 & 0 & 0 & 0 & 0 & \color{violet}\frac{1}{\sqrt{2}}\color{black} & 0 \\ 0 & 0 & 0 & \color{blue}\frac{1}{2}\color{black} & 0 & 0 & 0 & 0 & 0 & \color{blue}\frac{\sqrt{3}}{2}\color{black} \\ 0 & 0 & 0 & 0 & \color{red}\frac{2}{\sqrt{5}}\color{black} & 0 & 0 & 0 & 0 & 0 \\ 0 & 0 & 0 & 0 & 0 & \color{violet}\frac{1}{\sqrt{2}}\color{black} & 0 & 0 & 0 & 0 \\ \color{blue}-\frac{1}{2}\color{black} & 0 & 0 & 0 & 0 & 0 & \color{blue}\frac{\sqrt{3}}{2}\color{black} & 0 & 0 & 0 \end{bmatrix} ~.
    \end{aligned}
\end{equation}
As we can see, this Hamiltonian indeed requires two orders of squeezing to implement: 3rd-order squeezing to tune the matrix elements on the 3rd off-diagonal and 9th-order squeezing for the $-\frac{1}{2}$ matrix element at the anti-diagonal corners of the matrix.

\emph{Construction 2.}
Since this binomial code is the `minimal correcting unit' with respect to the error set $\mathcal{E}_{2}$, this is the same Hamiltonian we would get if we used the construction in Theorem~\ref{thm:SqueezingScaling} with desired ET distance set to $l = 2$.
However, if we instead only seek an ET distance of $l = 1$, Theorem~\ref{thm:SqueezingScaling} will give us a different Hamiltonian that only requires a single order of squeezing.
In particular, the parity manifold Hamiltonians will become:
\begin{equation}
    \begin{aligned}
        H_{0} & = \frac{1}{3}\begin{bmatrix} 0 & \sqrt{3} & 0 & 0 \\ \sqrt{3} & 0 & 2 & 0 \\ 0 & 2 & 0 & \sqrt{3} \\ 0 & 0 & \sqrt{3} & 0 \end{bmatrix} \\
        H_{1} & = \frac{1}{\sqrt{2}}\begin{bmatrix} 0 & 1 & 0 \\ 1 & 0 & 1 \\ 0 & 1 & 0 \end{bmatrix} \\
        H_{2} & = \begin{bmatrix} 0 & 0 & 0 \\ 0 & 0 & 0 \\ 0 & 0 & 0 \end{bmatrix} ~,
    \end{aligned}
    \label{eqn:exampleHnew}
\end{equation}
which are all nearest-neighbor, meaning the resulting Hamiltonian $H$ will only require 3rd-order squeezing.
This $H_{0}$ has changed by replacing the $|0_{\hat{n}}\rangle\langle 1_{\hat{n}}| + |1_{\hat{n}}\rangle\langle 0_{\hat{n}}|$ contribution in Eq.~\eqref{eqn:exampleH0} (which is no longer needed since $\hat{n} \not\in \mathcal{E}_{1}$) with some other coupling, orthogonal to the code space, that specifically cancels out the $(H_{0})_{0, 3} = (H_{0})_{3, 0}$ matrix element.
This $H_{1}$ has not changed because it is only a $3 \times 3$ matrix, meaning there are no other terms that can be added to the $H_{1} = |0_{\hat{a}}\rangle\langle 1_{\hat{a}}| + |1_{\hat{a}}\rangle\langle 0_{\hat{a}}|$ eigendecomposition to cancel out matrix elements -- in general, for codes with larger $K$, this will not be the case.
Finally, $H_{2} = 0$ since $\hat{a}^{2} \not\in \mathcal{E}_{1}$, so we do not need to impose any sort of evolution in the second error parity manifold to achieve the desired ET distance of $l = 1$.

We remark that, under this $l = 1$ Theorem~\ref{thm:SqueezingScaling} construction, $H_{0} = \frac{1}{3 / 2}\hat{J}_{x}^{(3 / 2)}$ and $H_{1} = \hat{J}_{x}^{(1)}$, where $\hat{J}_{x}^{(J)}$ is the spin-$J$ $\hat{J}_{x}$ operator.
Indeed, it holds more generally for any binomial code that the $l = 1$ Theorem~\ref{thm:SqueezingScaling} construction will yield $H_{0} = \frac{1}{J}\hat{J}_{x}^{(J)}$, $H_{1} = \frac{1}{J - 1/2}\hat{J}_{x}^{(J - 1/2)}$, and $H_{m} = 0$ for $m > 1$, where $J = \frac{K}{2}$ is the spin associated with the dimension $K + 1$ of the logical parity manifold.
This follows because, as first pointed out in~\cite{albert_performance_2018}, under the spaced Holstein-Primakoff map $|kN\rangle \leftrightarrow |J, J - k\rangle$, the $|\pm_{L}\rangle$ states of the binomial codes correspond to the spin-$J$ spin-coherent states lying on the $\pm\hat{x}$ axes, meaning they are the $\pm J$ eigenstates of the $\hat{J}_{x}$ spin operator under this mapping.
Thus, $\frac{1}{J}\hat{J}_{x}^{(J)}$ is a nearest-neighbor (so only requires a single order of squeezing) and Hermitian operator that generates continuous $\bar{X}(\theta)$ gates in the logical parity manifold.
Since the $l = 1$ Theorem~\ref{thm:SqueezingScaling} construction uniquely determines the logical parity manifold Hamiltonian, we therefore must have $H_{0} = \frac{1}{J}\hat{J}_{x}^{(J)}$.
Similarly, due to the combinatorial identity $k\binom{K}{k} = K\binom{K - 1}{k - 1}$, the $|\pm_{\hat{a}}\rangle \propto \left(\sum_{n = 0}^{KN}\sqrt{n}|n - 1\rangle\langle n|\right)|\pm_{L}\rangle$ error words are also spin-coherent states lying on the $\pm\hat{x}$ axes, but this time of spin $J' = \frac{K - 1}{2}$ since the $|0\rangle$ Fock state has been annihilated.
Since the $l = 1$ Theorem~\ref{thm:SqueezingScaling} construction also uniquely determines the first error parity manifold Hamiltonian (but not any of the other $H_{m}$ for $m > 1$), we then will also have $H_{1} = \frac{1}{J - 1/2}\hat{J}_{x}^{(J - 1/2)}$.

\emph{Construction 3.}
We now improve on this single squeezing result from Theorem~\ref{thm:SqueezingScaling} by using Theorem~\ref{thm:SingleSqueezing} to achieve error-transparency to $\mathcal{A}_{2} = \{\hat{I}, \hat{a}, \hat{a}^{2}\} \subseteq \mathcal{E}_{2}$, still only using a single order of squeezing.
In this case, this entails adding to Eq.~\eqref{eqn:exampleHnew} the $H_2$ from Eq.~\eqref{eqn:exampleH12} .
We note that, in general, the new Hamiltonians $H_1, H_{2}$ will be different from the ones from the Sec.~\ref{sec:derivingGates} construction.
We also note that, unlike $H_{0}$ and $H_{1}$, this second error parity manifold Hamiltonian will not generically be equal to a spin matrix, and so generally should be determined from our formula in Eq.~\eqref{eqn:thm2}.

We now demonstrate the performance of these three constructions by plotting their gate infidelities, following recovery, when subject to the pure photon-loss channel for a range of error rates $\kappa$ (Fig.~\ref{fig:performance}).
As expected, we find that \emph{Construction 1} (Sec.~\ref{sec:derivingGates} or $l = 2$ Theorem~\ref{thm:SqueezingScaling}), requiring two orders of squeezing, has infidelities scaling with $\gamma^{3}$, exactly matching the infidelities during idle evolution.
Likewise, we find that \emph{Construction 2} ($l = 1$ Theorem~\ref{thm:SqueezingScaling}), requiring only a single order of squeezing, has infidelities scaling with $\gamma^{2}$, meaning it indeed achieves an ET distance of $\lfloor 1 / 2 \rfloor + 1 = 1$.
Finally, we find that the improved single squeezing \emph{Construction 3} (Theorem~\ref{thm:SingleSqueezing}), while still scaling with $\gamma^{2}$ for small $\gamma$ since it is not ET to $\hat{n} \in \mathcal{E}_{2}$, yields about a factor of $10$ improvement over \emph{Construction 2} due to being additionally ET to the second photon-jump $\hat{a}^{2} \in \mathcal{A}_{2} \supset \mathcal{E}_{1}$.

\begin{figure*}
    \centering
    \includegraphics[width = \textwidth]{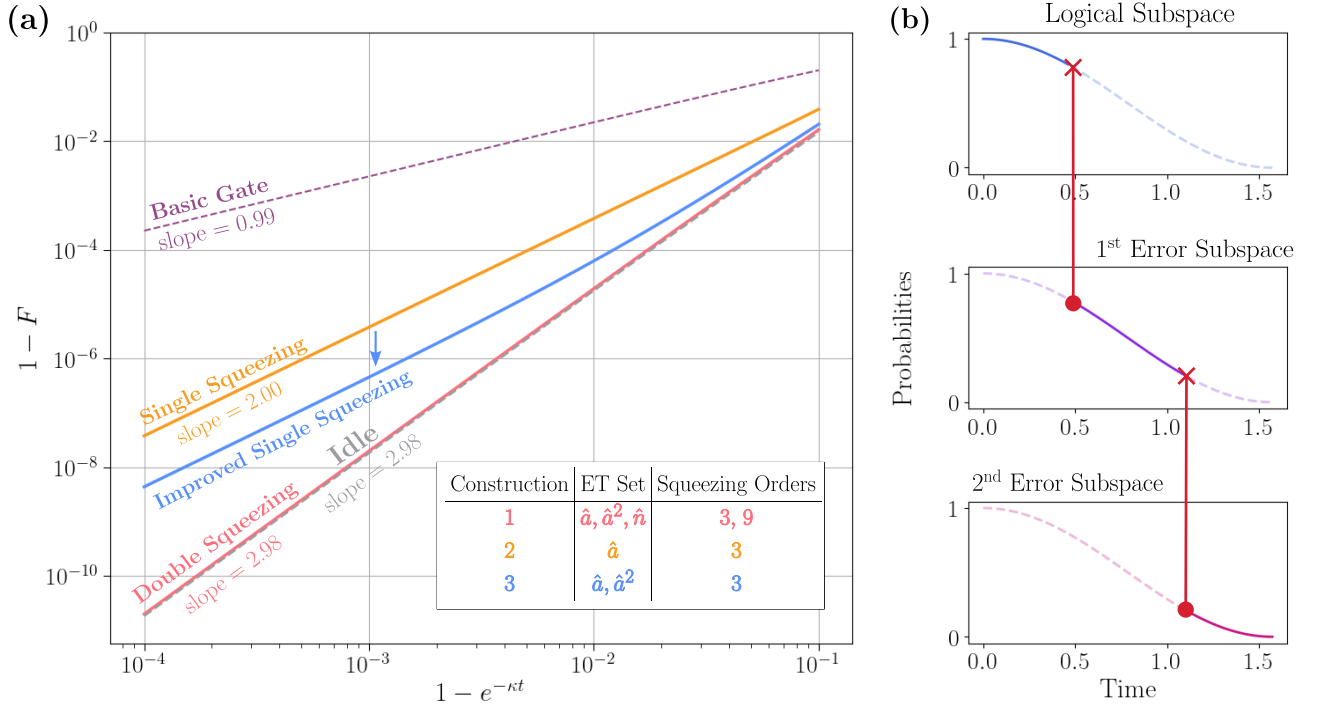}
    \caption{ 
        \textbf{Parity nested gate performance.} 
        \textbf{(a)} Average gate infidelities, following recovery, of various parity nested $\bar{X}$ gates for the $N = 3$, $K = 3$ binomial code, subject to the pure photon-loss channel with total gate time $t = \frac{\pi}{2}$ and a range of error rates $\kappa$ (see Appendix~\ref{app:performance}). The double squeezing (red), single squeezing (orange), and improved single squeezing (blue) gate Hamiltonians come from constructions 1, 2, and 3, respectively, outlined in Sec.~\ref{sec:performance}. The table inset summarizes the ET properties and required squeezing orders of these three constructions. For comparison, the performance of a basic $\bar{X}$ gate (purple), generated by $H = |0_{L}\rangle\langle 1_{L}| + |1_{L}\rangle\langle 0_{L}|$, and idle evolution (gray) are also shown. As can be seen, the improved single squeezing gate essentially matches the scaling of idle evolution (and double squeezing) for large error rates, and yields a constant $8.6$ factor improvement over the single squeezing infidelity for smaller error rates. This constant factor is not generic, decreasing for larger gate times $t$ and increasing for larger binomial codes (see Appendix~\ref{app:performance}). \textbf{(b)} Depiction of the time evolution in each of the error subspaces for the improved single squeezing gate. The solid lines show an example trajectory of the $|0_{L}\rangle$ state for one possible set of photon jump times.
    }
    \label{fig:performance} 
\end{figure*}

\section{Discussion}
We have introduced the concept of parity-nested operations and applied them to design error-transparent (ET) amplitude-mixing operations for the binomial encodings.
These gates can be ET to all the correctable photon jumps using only a single squeezing order, or can achieve any allowed ET distance $l$ using $\lfloor l / 2 \rfloor + 1$ squeezing orders.
This provides a conceptual completion of the ET gate set for binomial encodings, though a practical implementation remains to be designed.
Although presented here for the binomial codes, our theorems and constructions more generally apply to any hard Fock state cutoff rotation-symmetric code that can exactly correct for photon losses (which are in fact combinations of binomial codes; see Appendix~\ref{app:thm1}).
In particular, the same constructions can also be applied to yield ET $\bar{Y}(\theta)$ gates (or any logical rotation about equatorial axes of the Bloch sphere), since they are just $\bar{X}(\theta)$ gates for the rotation-symmetric code $\mathcal{C}'$ with codewords $|0_{L}\rangle' = |0_{L}\rangle$ and $|1_{L}\rangle' = i|1_{L}\rangle$.
Direct extension of our concept to other bosonic encodings that do not have a hard Fock state cutoff may require tuning an infinite number of couplings, but such ET operations could have perfect fidelity in principle.

The possibility of experimental implementation remains to be accomplished. 
We describe one possible implementation in Appendix~\ref{app:concept_implementation} using a version of our recently demonstrated matrix-element modification protocol applied to a squeezing drive rather than a linear displacement~\cite{roy_synthetic_2025}.
We also give a preliminary discussion of how this protocol can be made error-transparent to ancilla errors in Appendix~\ref{app:ancillaET}.
This proof-of-concept method shows that such parity-nested operations are possible in principle, but the strict adiabaticity requirements makes that protocol impractical.
Building on that protocol with optimal control methods~\cite{khaneja_optimal_2005,de_fouquieres_second_2011} may have the potential to successfully accomplish high-fidelity and highly ET gates in a short duration by building on the insights pointed out in our work.
We also remark that squeezing \cite{wang_efficient_2020}, trisqueezing \cite{eriksson_universal_2024} and six wave mixing processes \cite{vanselow_dissipating_2025} have been demonstrated in experiments which can be leveraged for an experimental implementation of ET gates.
Although experimentally implementing our specific Hamiltonians may be challenging, our work provides a conceptual touchstone through which future efforts can be guided, motivating working towards either (a) efficiently generating these high orders of squeezing, (b) searching for approximate or numerically-optimized ET operations based on our constructions, or (c) exploring new bosonic encodings which may be more compatible with amplitude-mixing ET gates.

Finally, we remark that our ET constructions are also compatible with qudit encodings that allow for multi-frequency control. 
In those systems, multi-frequency control makes possible tuning of effective matrix elements in a multi-frequency rotating wave approximation.
The relevant experimental systems include heavy transmons~\cite{neeley_emulation_2009,wu_high-fidelity_2020,fischer_universal_2023,wang_high-e_je_c_2025}, nuclear spins in solid state defects~\cite{asaad_coherent_2020,gupta_robust_2024,fernandez_de_fuentes_navigating_2024,vaartjes_certifying_2025,yu_schrodinger_2025}, and the multi-component manifolds of trapped ions and neutral atoms~\cite{ringbauer_universal_2022,aksenov_realizing_2023,hrmo_native_2023,buchemmavari_entangling_2024,omanakuttan_quantum_2021,anderson_accurate_2015}. 
To this end, although our ET constructions focus on errors generated by photon loss jump operators, Table~\ref{tab:errors} shows how the same operations are ET against dephasing errors, which are the dominant error in some of these cases.

\section*{Acknowledgements and Other Information}

\paragraph{Author contributions}
O. C. W. led the research.
S. R. accomplished the simulations shown in App.~\ref{app:concept_implementation}.
B. R. and V. F. supervised and guided the effort, including identification of the central question.
O. C. W. wrote the manuscript, with input from all authors.

\paragraph{Funding information}
O. C. W., S. R., and V. F. are thankful for the support of the Aref and Manon Lahham Faculty Fellowship that contributed to this work.
B. R. acknowledges funding from the Canada First Research Excellence Fund, the Natural Sciences and Engineering Research Council of Canada (NSERC) as well as the Fonds de Recherche du Québec, Nature et Technologie (FRQNT).

\paragraph{Other acknowledgements}
We acknowledge illuminating discussions with Steve Yu and Andrea Morello regarding nuclear spin cats in solid state defects. 

\paragraph{Data and code availability}
All data generated and code used in this work are available at: \href{https://doi.org/10.5281/zenodo.15491315}{10.5281/zenodo.15491315}.

\paragraph{Competing Interests}
The authors declare no competing interests.

\bibliographystyle{quantum}
\bibliography{vf-references.bib, ow-references.bib}

\appendix
\numberwithin{equation}{section}

\section{Error-Transparent Phase Gates for Binomial Codes} \label{app:phasegates}
As mentioned in Sec.~\ref{sec:ETgates}, error-transparent (ET) phase gates can be accomplished using diagonal Hamiltonians of the form $H = \sum_{n = 0}^{n_{\text{trc}}}\varphi_{n}|n\rangle\langle n|$.
In particular, since the $|0_{L}\rangle$ and $|1_{L}\rangle$ states of rotation-symmetric codes have support on disjoint Fock states (see Eq.~\eqref{eqn:rotSymmetric}), setting $\varphi_{n} = 0$ for $|n\rangle$ in the support of $|0_{L}\rangle$ and $\varphi_{n} = \theta$ for $|n\rangle$ in the support of $|1_{L}\rangle$ automatically yields a phase gate $\bar{Z}(\theta)$ on the logical subspace: $\exp(-iH)|0_{L}\rangle = |0_{L}\rangle$ and $\exp(-iH)|0_{L}\rangle = \exp(-i\theta)|1_{L}\rangle$.
Since the error words $|0_{\hat{a}}\rangle \propto \hat{a}|0_{L}\rangle$ and $|1_{\hat{a}}\rangle \propto \hat{a}|1_{L}\rangle$ also have disjoint Fock support, we can likewise set the $\varphi_{n}$ for the $|n\rangle$ in this error subspace (which are disjoint from the $|n\rangle$ in the logical subspace) so that $H$ also generates a phase gate $\bar{Z}(\theta)$ in the $\hat{a}$ error subspace. The resulting $H$ will then satisfy the ET conditions (Eq.~\eqref{eqn:ETCondition}) for the error set $\mathcal{A}_{1} = \{\hat{I}, \hat{a}\}$.
This construction can easily be extended to more photon losses $\mathcal{A}_{l} = \{\hat{I}, \hat{a}, \dots, \hat{a}^{l}\}$ (as long as $\mathcal{A}_{l}$ is correctable) by repeating this process for each $\hat{a}^{m}$ error subspace, all of which live on disjoint Fock states.
Since $H$ is diagonal in the Fock basis, it will also automatically commute with the $\hat{n}^{k}$ no-jump contributions of the errors in $\mathcal{E}_{l}$ (Eq.~\eqref{eqn:ETSet}), implying that these phase gates are more generally ET to $\mathcal{E}_{l}$.
Thus, phase gates with a saturated ET distance (i.e. matching the order of idle protection) can be achieved for binomial codes with relatively simple diagonal Hamiltonians.

We remark that several parts of this error-transparency construction rely on binomial phase gates not requiring the mixing of different Fock states (i.e. $H$ can be diagonal).
In particular, diagonal matrix elements interact simply under commutation with the $\sqrt{n}$ factor of $\hat{a}$, and the $\hat{n}^{k}$ no-jump contributions, unlike their non-diagonal counterparts that are essential for amplitude-mixing gates which are the focus of our work.

\section{Parity Nested Gate Constructions}
\subsection{Error-transparency to $\mathcal{E}_{l}$ (Theorem 1)}\label{app:thm1}
We start with the proof of Theorem~\ref{thm:SqueezingScaling}, which we repeat here:
\newtheorem*{T1}{Theorem~\ref{thm:SqueezingScaling}}
\begin{T1}
  For any rotation-symmetric code with a finite number of nonzero coefficients $\vec{c} = [c_{0}, \dots, c_{K}]$ which exactly corrects the error set $\mathcal{E}_{l}=\left\{\hat{a}^{m}\hat{n}^{k} \, \big| \, \frac{m}{2} + k \leq \frac{l}{2}\right\}$, we can always construct a parity nested Hamiltonian $H$ generating continuous logical $\bar{X}(\theta)$ gates which is error-transparent to the error set $\mathcal{E}_{l}$ and only requires $\left\lfloor l / 2 \right\rfloor + 1$ orders of squeezing  to implement.
\end{T1}
As discussed in the main text, this construction is based on breaking down the coefficients $[c_{0}, \dots, c_{K}]$ defining the code into the coefficients of the `minimal encoding unit' binomial code which corrects the error set $\mathcal{E}_{l}$.
In particular, we use the following lemma:
\begin{lemma}
  A rotation-symmetric code defined by coefficients $\vec{c} = [c_{0}, \dots, c_{K}]$ exactly corrects the error set $\mathcal{E}_{l}=\left\{\hat{a}^{m}\hat{n}^{k} \, \big| \, \frac{m}{2} + k \leq \frac{l}{2}\right\}$ if and only if the vector of squared coefficients $\left[|c_{0}|^{2}, \dots, |c_{K}|^{2}\right]$ can be written as a linear combination of the set of variably shifted binomial coefficient vectors
  \begin{equation}
    \begin{aligned}
      & \left\{\left[\binom{K'}{0}, \binom{K'}{1}, \dots, \binom{K'}{K'}, 0, \dots, 0\right],\right. \\
      & \left[0, \binom{K'}{0}, \binom{K'}{1}, \dots, \binom{K'}{K'}, 0, \dots, 0\right],  \\
      & \left. \dots, \left[0, \dots, 0, \binom{K'}{0}, \binom{K'}{1}, \dots, \binom{K'}{K'}\right]\right\} ~,
    \end{aligned}
    \label{eqn:binCoeffVecs}
  \end{equation}
  where $K' = l + 1$.
\end{lemma}
\begin{proof}
A hard energy cutoff order-$N$ rotation-symmetric code exactly corrects the error set $\mathcal{E}_{l}$ if and only if the $\hat{n}^{l'}$ moments are equal for the $|0_{L}\rangle$ and $|1_{L}\rangle$ codewords up to $\hat{n}^{l}$~\cite{michael_new_2016}.
In other words,
\begin{equation}
  \begin{aligned}
    \langle 0_{L} | \hat{n}^{l'} | 0_{L} \rangle = \langle 1_{L} | \hat{n}^{l'} | 1_{L} \rangle ~,
  \end{aligned}
  \label{eqn:equalMoments}
\end{equation}
for all $l' \in \{0, \dots, l\}$, where the $l' = 0$ case ensures the code words are normalized.
This set of conditions can be written as
\begin{equation}
  \begin{aligned}
    \sum_{k \text{ even}}^{K}c_{k}^{*}(Nk)^{l'}c_{k} & = \sum_{k \text{ odd}}^{K}c_{k}^{*}(Nk)^{l'}c_{k} \\
    \iff \sum_{k = 0}^{K}(-1)^{k}k^{l'}|c_{k}|^{2} & = 0 ~,
  \end{aligned}
  \label{eqn:squaredCoeffsCond}
\end{equation}
for all $l' \in \{0, \dots, l\}$, which is notably a linear system of equations in the vector of squared coefficients $\left[|c_{0}|^{2}, \dots, |c_{K}|^{2}\right]$.
Since we know the $K' = l + 1$ binomial code corrects the error set $\mathcal{E}_{l}$, its coefficients (given by the square root of the binomial coefficients) must satisfy Eq.~\eqref{eqn:squaredCoeffsCond}, in particular implying $\left[\binom{K'}{0}, \binom{K'}{1}, \dots \binom{K'}{K'}, 0, \dots, 0\right]$ is one solution to this linear system.
Since the binomial code with cutoff $K' + 1$ also corrects the error set $\mathcal{E}_{l}$, $\left[\binom{K' + 1}{0}, \binom{K' + 1}{1}, \dots \binom{K' + 1}{K' + 1}, 0, \dots, 0\right]$ is also a solution to Eq.~\eqref{eqn:squaredCoeffsCond}, immediately implying by linearity and Pascal's rule that
\begin{widetext}
    \begin{equation}
      \begin{aligned}
        \left[0, \binom{K'}{0}, \binom{K'}{1}, \dots \binom{K'}{K'}, 0, \dots, 0\right] = & \left[\binom{K' + 1}{0}, \binom{K' + 1}{1}, \dots \binom{K' + 1}{K' + 1}, 0, \dots, 0\right] - \\
        & \left[\binom{K'}{0}, \binom{K'}{1}, \dots \binom{K'}{K'}, 0, \dots, 0\right]
      \end{aligned}
    \end{equation}
    is also a solution to Eq.~\eqref{eqn:squaredCoeffsCond}.
    
    We can recursively repeat this process to shift the binomial coefficients to the right by any amount.
    In particular, since the binomial code with cutoff $K' + J$ also corrects the error set $\mathcal{E}_{l}$ for any $J \geq 0$, we can conclude by linearity and standard combinatorial identities that
    \begin{equation}
      \begin{aligned}
        \left[0_{1}, \dots, 0_{J}, \binom{K'}{0}\right. & \left., \binom{K'}{1}, \dots \binom{K'}{K'}, 0_{1}, \dots, 0_{K - K' - J}\right] \\
        ={}& \left[\binom{K' + J}{0}, \binom{K' + J}{1}, \dots \binom{K' + J}{K' + J}, 0_{1}, \dots, 0_{K - K' - J}\right] - \\
        &\sum_{j = 0}^{J - 1}a_{j}\left[0_{1}, 0_{2}, \dots, 0_{j}, \binom{K'}{0}, \binom{K'}{1}, \dots \binom{K'}{K'}, 0_{1}, \dots, 0_{K - K' - J}\right]
      \end{aligned}
      \label{eqn:shiftedBinCoeffsProof}
    \end{equation}
\end{widetext}
solves Eq.~\eqref{eqn:squaredCoeffsCond} if all the less-shifted vectors in this summation already do, where $a_{j}$ are some coefficients, which happen to be equal to $a_{j} = \binom{J}{j}$.
This Eq.~\eqref{eqn:shiftedBinCoeffsProof} follows from the fact that Pascal's rule implies the $(K' + 1)$-th row of Pascal's triangle is equal to the element-wise sum of two copies of the $K'$-th row, which applied recursively in turn implies the $(K' + J)$-th row is equal to some linear combination of different shifted versions of the $K'$-th row -- since $\binom{K' + J}{K' + J} = 1$ this linear combination must only involve one copy of the `most shifted' $K'$-th row (i.e. the vector on the LHS of Eq.~\eqref{eqn:shiftedBinCoeffsProof}).

Since we already know the $J = 0$ case (with no initial padding zeros) solves Eq.~\eqref{eqn:squaredCoeffsCond}, we can therefore conclude by strong induction on $J$ that all the variably shifted binomial coefficient vectors, given in Eq.~\eqref{eqn:binCoeffVecs}, also solve Eq.~\eqref{eqn:squaredCoeffsCond} and thus correspond to codes which exactly correct the error set $\mathcal{E}_{l}$.
Once again by linearity, this means that if the vector of squared coefficients $\left[|c_{0}|^{2}, \dots, |c_{K}|^{2}\right]$ can be written as a linear combination of the variably shifted binomial coefficient vectors, then it will satisfy Eq.~\eqref{eqn:equalMoments} and thus exactly correct the error set $\mathcal{E}_{l}$, proving one direction of the implication in the lemma.
The other direction of the implication follows because each of the $l + 1$ linear equations in Eq.~\eqref{eqn:squaredCoeffsCond} are independent, and there are $K + 1$ variables in $\left[|c_{0}|^{2}, |c_{1}|^{2}, \dots, |c_{K}|^{2}\right]$, so the solution subspace has dimension $K + 1 - (l + 1) = K - l$.
There are $K - l$ vectors in Eq.~\eqref{eqn:binCoeffVecs} and they are linearly independent (since when stacked vertically, they yield a matrix that is automatically in row echelon form with $K - l$ pivots), so they must span the solution subspace, meaning any vector of squared coefficients $\left[|c_{0}|^{2}, |c_{1}|^{2}, \dots, |c_{K}|^{2}\right]$ satisfying Eq.~\eqref{eqn:squaredCoeffsCond}, and thus corresponding to a code which exactly corrects $\mathcal{E}_{l}$, must be some linear combination of the variably shifted binomial coefficient vectors.
\end{proof}

As discussed in the main text, the idea of the proof of Thoerem~\ref{thm:SqueezingScaling} is to use this lemma to decompose the code into a bunch of smaller shifted copies of the cutoff $K' = l + 1$ binomial code, apply Eq.~\eqref{eqn:HETgeneral} to each of these copies of the binomial code (each of which requires $\left\lfloor l / 2 \right\rfloor + 1$ off-diagonals to achieve error-transparency), and then recombine the series of resulting Hamiltonians to form a Hamiltonian that is error-transparent (ET) to the original code and still only utilizes $\left\lfloor l / 2 \right\rfloor + 1$ off-diagonals.
The difficulty here lies in the fact that the lemma deals with the square of the coefficients, not the coefficients themselves, so this decomposition is not a simple linear combination, and therefore this recombination process of forming the final ET Hamiltonian is not entirely straightforward.
Nevertheless, this general procedure still works, and so we now give the proof of Theorem~\ref{thm:SqueezingScaling}, repeated again for convenience, based on this idea.
\begin{T1}
  For any rotation-symmetric code with a finite number of nonzero coefficients $\vec{c} = [c_{0}, \dots, c_{K}]$ which exactly corrects the error set $\mathcal{E}_{l}=\left\{\hat{a}^{m}\hat{n}^{k} \, \big| \, \frac{m}{2} + k \leq \frac{l}{2}\right\}$, we can always construct a parity nested Hamiltonian $H$ generating continuous logical $\bar{X}(\theta)$ gates which is error-transparent to the error set $\mathcal{E}_{l}$ and only requires $\left\lfloor l / 2 \right\rfloor + 1$ orders of squeezing  to implement.
\end{T1}

\begin{proof}
We start with the simpler case of decomposing a given order-$N$ hard energy cutoff rotation-symmetric code $\mathcal{C}$ into two order-$N$ rotation-symmetric codes $\mathcal{A}$ and $\mathcal{B}$ which each also exactly correct $\mathcal{E}_{l}$ and have coefficients given by $\vec{a} = [a_{0}, \dots, a_{K}]$ and $\vec{b} = [b_{0}, \dots, b_{K}]$, respectively.
In particular, following Lemma 1, this decomposition is in terms of the square of each coefficient and says that $|c_{k}|^{2} = \alpha|a_{k}|^{2} + \beta|b_{k}|^{2}$ for all $k \in \{0, \dots, K\}$.
In other words, following Eq.~\eqref{eqn:dualBasisRotSymmetric}, we can write this in bra-ket notation as
\begin{widetext}
    \begin{equation}
        \begin{aligned}
            |\pm_{c}\rangle = \sum_{k = 0}^{K}(\pm1)^{k}e^{i\varphi_{k}}\sqrt{\alpha|\langle kN | \pm_{a}\rangle|^{2} + \beta|\langle kN | \pm_{b}\rangle|^{2}}|kN\rangle ~,
        \end{aligned}
        \label{eqn:codewordConstruction}
    \end{equation}
\end{widetext}
where $|\pm_{a / b / c}\rangle$ are the $|\pm_{L}\rangle$ states of codes $\mathcal{A}$, $\mathcal{B}$, and $\mathcal{C}$, respectively, and $\varphi_{k} = \text{arg}(c_{k})$.
Since the codes $\mathcal{A}$ and $\mathcal{B}$ exactly correct $\mathcal{E}_{l}$, we can use the eigendecomposition construction outlined in Sec.~\ref{sec:derivingGates} to construct Hamiltonians $H^{a}$ and $H^{b}$ which generate continuous logical $\bar{X}$ gates on $\mathcal{A}$ and $\mathcal{B}$, respectively, and are each ET to $\mathcal{E}_{l}$.
We then construct the new Hamiltonian $H^{c}$ one parity manifold at a time according to
\begin{equation}
  \begin{aligned}
    (H_{m}^{c})_{k, k'} = \frac{\alpha (H_{m}^{a})_{k, k'}a_{k}^{*}a_{k'} + \beta (H_{m}^{b})_{k, k'}b_{k}^{*}b_{k'}}{c_{k}^{*}c_{k'}} ~,
  \end{aligned}
  \label{eqn:HConstruction}
\end{equation}
where, critically, we see that if $(H_{m}^{a})_{k, k'} = (H_{m}^{b})_{k, k'} = 0$, then $(H_{m}^{c})_{k, k'} = 0$, meaning $H^{c}$ does not make use of any new matrix elements that were not used before by either $H^{a}$ or $H^{b}$.
We note that if $c_{k}$ or $c_{k'}$ is zero, then when using Eq.~\eqref{eqn:HConstruction}, any arbitrary (but consistent) nonzero values can be plugged in for them instead (e.g. if $c_{4} = c_{7} = 0$, then values of $c_{4} = 1$ and $c_{7} = 2$ could be plugged in whenever $c_{4}$ or $c_{7}$ appear in Eq.~\eqref{eqn:HConstruction}).
In this case, the construction will still work and all the following reasoning (including Eq.~\eqref{eqn:gateDerivation} and Eq.~\eqref{eqn:ETDerivation}) will still hold.

We now show that this new Hamiltonian $H^{c} = \sum_{m = 0}^{N - 1}H_{m}^{c}$ generates continuous logical $\bar{X}(\theta)$ gates on $\mathcal{C}$ and is also ET to $\mathcal{E}_{l}$.
For this first continuous $\bar{X}(\theta)$ gate part, we indeed have:
\begin{widetext}
    \begin{equation}
      \begin{aligned}
        H^{c}|\pm_{c}\rangle & = (H_{0}^{c})|\pm_{c}\rangle \\
        & = \sum_{k = 0}^{K}\left(\sum_{k' = 0}^{K}(H_{0}^{c})_{k, k'}(\pm1)^{k'}c_{k'}\right)|kN\rangle \\
        & = \sum_{k = 0}^{K}\left(\sum_{k' = 0}^{K}\frac{\alpha (H_{0}^{a})_{k, k'}a_{k}^{*}a_{k'} + \beta (H_{0}^{b})_{k, k'}b_{k}^{*}b_{k'}}{c_{k}^{*}c_{k'}}(\pm1)^{k'}c_{k'}\right)|kN\rangle \\
        & = \sum_{k = 0}^{K}\frac{1}{c_{k}^{*}}\left(\alpha a_{k}^{*}\sum_{k' = 0}^{K}(H_{0}^{a})_{k, k'}(\pm1)^{k'}a_{k'} + \beta b_{k}^{*}\sum_{k' = 0}^{K}(H_{0}^{b})_{k, k'}(\pm1)^{k'}b_{k'}\right)|kN\rangle \\
        & = \sum_{k = 0}^{K}\frac{1}{c_{k}^{*}}\left(\alpha a_{k}^{*}\left(\pm(\pm)^{k}a_{k}\right) + \beta b_{k}^{*}\left(\pm(\pm)^{k}b_{k}\right)\right)|kN\rangle \\
        & = \pm\sum_{k = 0}^{K}\frac{(\pm)^{k}}{c_{k}^{*}}(\alpha|a_{k}|^{2} + \beta|b_{k}|^{2})|kN\rangle \\
        & = \pm\sum_{k = 0}^{K}(\pm)^{k}c_{k}|kN\rangle \\
        & = \pm|\pm_{c}\rangle ~,
      \end{aligned}
      \label{eqn:gateDerivation}
    \end{equation}
    meaning $H^{c}$ generates continuous logical $\bar{X}(\theta)$ gates on $\mathcal{C}$, where we have used the fact that $(H_{0}^{a})|\pm_{a}\rangle = \pm|\pm_{a}\rangle$ implies $\sum_{k' = 0}^{K}(H_{0}^{a})_{k, k'}(\pm1)^{k'}a_{k'} = \pm(\pm)^{k}a_{k}$, and likewise for $\mathcal{B}$.
    
    Now, for the error-transparency part, consider an arbitrary error $\hat{a}^{m}\hat{n}^{j} \in \mathcal{E}_{l}$ with $m \leq l \leq N - 1$ and $j \leq \left\lfloor \frac{l - m}{2} \right\rfloor$.
    Since $H^{a}$ and $H^{b}$ are both ET to $\mathcal{E}_{l}$, we have $\left[H^{a}, \hat{a}^{m}\hat{n}^{j}\right]|\pm_{a}\rangle = 0 \Rightarrow (H_{m}^{a})\hat{a}^{m}\hat{n}^{j}|\pm_{a}\rangle = \hat{a}^{m}\hat{n}^{j}(H_{0}^{a})|\pm_{a}\rangle = \pm\hat{a}^{m}\hat{n}^{j}|\pm_{a}\rangle$, meaning $\sum_{k' = 0}^{K}(H_{m}^{a})_{k, k'}\left((\pm1)^{k'}\varepsilon_{k'}^{m, j}a_{k'}\right) = \pm(\pm)^{k}\varepsilon_{k}^{m, j}a_{k'}$, and likewise for $\mathcal{B}$, where $\varepsilon_{k}^{m, j} = \sqrt{(k'N - m + 1)\cdots(k'N - 1)(k'N)}(k'N)^{j}$ is the prefactor on the $k$-th coefficient induced by the error $\hat{a}^{m}\hat{n}^{j}$.
    We can then follow a similar line of reasoning as in Eq.~\eqref{eqn:gateDerivation} to prove that $H^{c}$ is also error-transparent to this arbitrary error $\hat{a}^{m}\hat{n}^{j}$. In particular, we have:
    \begin{equation}
      \begin{aligned}
        H^{c}\hat{a}^{m}\hat{n}^{j}|\pm_{c}\rangle & = (H_{m}^{c})\hat{a}^{m}\hat{n}^{j}|\pm_{c}\rangle \\
        & = \sum_{k = 0}^{K}\left(\sum_{k' = 0}^{K}(H_{m}^{c})_{k, k'}(\pm1)^{k'}\varepsilon_{k'}^{m, j}c_{k'}\right)|kN\rangle \\
        & = \sum_{k = 0}^{K}\frac{1}{c_{k}^{*}}\left(\alpha a_{k}^{*}\sum_{k' = 0}^{K}(H_{m}^{a})_{k, k'}(\pm1)^{k'}\varepsilon_{k'}^{m, j}a_{k'} + \beta b_{k}^{*}\sum_{k' = 0}^{K}(H_{m}^{b})_{k, k'}(\pm1)^{k'}\varepsilon_{k'}^{m, j}b_{k'}\right)|kN\rangle \\
        & = \pm\sum_{k = 0}^{K}\frac{(\pm)^{k}\varepsilon_{k}^{m, j}}{c_{k}^{*}}(\alpha|a_{k}|^{2} + \beta|b_{k}|^{2})|kN\rangle \\
        & = \pm\hat{a}^{m}\hat{n}^{j}|\pm_{c}\rangle \\
        & = \hat{a}^{m}\hat{n}^{j}H^{c}|\pm_{c}\rangle ~,
      \end{aligned}
      \label{eqn:ETDerivation}
    \end{equation}
\end{widetext}
meaning $[H^{c}, \hat{a}^{m}\hat{n}^{j}]|\pm_{c}\rangle = 0$, making $H^{c}$ error-transparent to $\hat{a}^{m}\hat{n}^{j}$, and thus also to the entire error set $\mathcal{E}_{l}$ (since this holds for any $\hat{a}^{m}\hat{n}^{j} \in \mathcal{E}_{l}$).
We remark that this part of the proof relies critically on the fact that all the errors in $\mathcal{E}_{l}$ only contain one nonzero matrix element in each row, which is what allows us to write the action of each error $\hat{a}^{m}\hat{n}^{j}$ as the simple accumulation of a prefactor $\varepsilon_{k}^{m, j}$ on each coefficient, which can then be distributed to each term of our decomposition in Eq.~\eqref{eqn:HConstruction}.
The construction given by Eq.~\eqref{eqn:HConstruction} thus indeed takes two order-$N$ rotation-symmetric codes $\mathcal{A}$ and $\mathcal{B}$, with coefficients $\vec{a} = [a_{0}, \dots, a_{K}]$ and $\vec{b} = [b_{0}, \dots, b_{K}]$ and logical $\bar{X}(\theta)$-generating and error-transparent (to $\mathcal{E}_{l}$) Hamiltonians $H^{a}$ and $H^{b}$, respectively, and outputs a new logical $\bar{X}(\theta)$-generating and error-transparent (to $\mathcal{E}_{l}$) Hamiltonian $H^{c}$ for a new order-$N$ rotation-symmetric code $\mathcal{C}$ whose coefficients $\vec{c} = [c_{0}, \dots, c_{K}]$ satisfy $|c_{k}|^{2} = \alpha|a_{k}|^{2} + \beta|b_{k}|^{2}$ for all $k \in \{0, \dots, K\}$ and some $\alpha$ and $\beta$.

Now consider the rotation-symmetric code from the statement of Theorem~\ref{thm:SqueezingScaling}.
From Lemma 1, there must exist some decomposition of the vector of its squared coefficients $[|c_{0}|^{2}, \dots, |c_{K}|^{2}]$ in terms of the squared coefficients of the variably shifted cutoff $K' = l + 1$ binomial codes.
First using Eq.~\eqref{eqn:HETgeneral} to construct Hamiltonians that generate continuous logical $\bar{X}(\theta)$ gates and are ET to $\mathcal{E}_{l}$ for each of these variably shifted codes, we can then recursively apply Eq.~\eqref{eqn:HConstruction} to combine these Hamiltonians based on the decomposition of the squared coefficients $[|c_{0}|^{2}, \dots, |c_{K}|^{2}]$ to eventually yield our desired error-transparent (to $\mathcal{E}_{l}$) Hamiltonian $H$ acting on the original code.
As we remarked below Eq.~\eqref{eqn:HConstruction}, this construction does not introduce any new matrix elements, so since each composite shifted cutoff $K' = l + 1$ binomial code Hamiltonian only contains nonzero matrix elements on the first $\left\lfloor \frac{K' + 1}{2} \right\rfloor = \left\lfloor l / 2 \right\rfloor + 1$ odd off-diagonals (see Sec.~\ref{sec:derivingGates}), $H$ must also only contain nonzero matrix elements on these same off-diagonals.
Thus, $H$ indeed only requires $\left\lfloor l / 2 \right\rfloor + 1$ orders of squeezing to implement.
\end{proof}

We remark that, although the construction in Eq.~\eqref{eqn:HConstruction} is useful for this proof, once the existence of these error-transparent Hamiltonians is established, there is a more straightforward way to compute them.
In particular, as is done in the proof of Theorem~\ref{thm:SingleSqueezing} in the following section, we can directly solve the linear system whose variables are the nonzero matrix elements of $H$ (on the first $\left\lfloor l / 2 \right\rfloor + 1$ odd off-diagonals) and whose constraining equations are given by the ET conditions in Eq.~\eqref{eqn:ECConditions} with error set $\mathcal{E} = \mathcal{E}_{l}$.
Therefore, the role of the above proof is to guarantee that this overconstrained linear system will always be consistent for any rotation-symmetric code, which is a nontrivial result.

\subsection{Error-transparency to $\mathcal{A}_{l}$ (Theorem 2)}\label{app:thm2}
We now give the derivation of Theorem~\ref{thm:SingleSqueezing}, which we repeat here for convenience:
\newtheorem*{T2}{Theorem~\ref{thm:SingleSqueezing}}
\begin{T2}
  For any rotation-symmetric code with a finite number of nonzero coefficients $\vec{c} = [c_{0}, \dots, c_{K}]$ which exactly corrects the error set $\mathcal{A}_{l} = \{\hat{I}, \hat{a}, \dots, \hat{a}^{l}\}$, consider the parity nested Hamiltonian $H$ formed from nearest-neighbor parity manifold Hamiltonians $H_{m}$ given by
  \begin{equation}
      \begin{aligned}
          (H_{m})_{k, k + 1} & = (H_{m})_{k + 1, k}^{*} \\
          & = \frac{1}{(c_{m})_{k}^{*}(c_{m})_{k + 1}}\sum_{j = 0}^{K}(-1)^{j + k}|(c_{m})_{j}|^{2} ~,
      \end{aligned}
      \label{eqn:thm2App}
  \end{equation}
  where the error coefficients $(c_{m})_{k} = \mathcal{N}_{m}\sqrt{(kN - m + 1)\cdots(kN)}c_{k}$ are defined below Eq.~\eqref{eqn:errorWordsRotSym}. The resulting $H$ then generates continuous logical $\bar{X}(\theta)$ gates, is error-transparent to the error set $\mathcal{A}_{l}$, and only requires a single order of squeezing to implement.
\end{T2}
\begin{proof}
    For a parity nested Hamiltonian $H$ to only require a single order of squeezing, it must only have a single nonzero off-diagonal, meaning each sub-matrix $H_{m}$ must only contain nearest-neighbor couplings. We thus start in the logical parity manifold by constructing -- for any order-$N$ rotation-symmetric code with coefficients $\vec{c}$ -- the unique nearest-neighbor logical subspace Hamiltonian $H_{0}$ which generates continuous $\bar{X}(\theta)$ gates.
    In particular, to ensure $H_{0}$ is Hermitian, nearest-neighbor, and generates the $\bar{X}(\theta)$ gate, we must choose the $(H_{0})_{k, k + 1} = (H_{0})_{k + 1, k}^{*}$ nearest-neighbor matrix elements (for $0 \leq k < K$) which solve the eigenstate equation
    \begin{widetext}
        \begin{equation}
            \begin{aligned}
                H_{0}|\pm_{L}\rangle & = \pm|\pm_{L}\rangle \\
                \iff \left((\pm)^{k - 1}(H_{0})_{k - 1, k}^{*}c_{k - 1} + (\pm)^{k + 1}(H_{0})_{k, k + 1}c_{k + 1}\right)|kN\rangle & = \pm (\pm)^{k}c_{k}|kN\rangle \hspace{10pt} \forall 0 \leq k \leq K \\
                \iff (H_{0})_{k - 1, k}^{*}c_{k - 1} + (H_{0})_{k, k + 1}c_{k + 1} & = c_{k} \hspace{10pt} \forall 0 \leq k \leq K ~,
            \end{aligned}
            \label{eqn:eigenstateEqn}
        \end{equation}    
    \end{widetext}
    where $(H_{0})_{-1, 0} = (H_{0})_{K, K + 1} = 0$, so that the first and last of these equations only have one term on the left hand side.
    Conveniently, we see that all the $\pm$ factors cancel out, yielding the exact same set of $K + 1$ equations for both the $|+_{L}\rangle$ and $|-_{L}\rangle$ eigenstate equations.
    This is an over-constrained system of $K + 1$ real-linear equations in the $K$ variables $(H_{0})_{k, k + 1}$ (for $0 \leq k < K$), which is consistent (as it must be from the Theorem~\ref{thm:SqueezingScaling}) exactly when the coefficients $\vec{c}$ are normalized (Eq.~\eqref{eqn:normalization}).
    The unique solution to Eq.~\eqref{eqn:eigenstateEqn} is then given explicitly by
    \begin{equation}
        \begin{aligned}
            (H_{0})_{k, k + 1} = \frac{1}{c_{k}^{*}c_{k + 1}}\sum_{j = 0}^{k}(-1)^{j + k}|c_{j}|^{2} ~.
        \end{aligned}
        \label{eqn:H0Sol}
    \end{equation}
    
    We can similarly construct the $m$-th error parity manifold Hamiltonian $H_{m}$ by enforcing error-transparency to the corresponding error $\hat{a}^{m} \in \{\hat{a}, \dots, \hat{a}^{l}\}$.
    In particular, since the error words $|\pm_{\hat{a}^{m}}\rangle \propto \hat{a}^{m}|\pm_{L}\rangle$ only live in the $m$-th error parity manifold, we see that $H_{m}$ must satisfy the eigenstate equation $H_{m}|\pm_{\hat{a}^{m}}\rangle = \pm|\pm_{\hat{a}^{m}}\rangle$.
    Further requiring $H_{m}$ to be Hermitian and nearest-neighbor results in the exact same equations as in Eq.~\eqref{eqn:eigenstateEqn}, except with $|\pm_{L}\rangle$ replaced with $|\pm_{\hat{a}^{m}}\rangle$ and thus $c_{k}$ replaced with the error coefficients $(c_{m})_{k} = \mathcal{N}_{m}\sqrt{(kN)(kN - 1)\cdots(kN - m)}c_{k}$.
    This set of equations is then consistent exactly when $\sum_{k = 0}^{K}(-1)^{k}|(c_{m})_{k}|^{2} = 0$, which is exactly the statement that the order-$N$ rotation-symmetric code defined by $\vec{c}$ exactly corrects the error $\hat{a}^{m}$.
    The unique $m$-th error subspace Hamiltonian $H_{m}$ is then given explicitly by
    \begin{equation}
        \begin{aligned}
            (H_{m})_{k, k + 1} = \frac{1}{(c_{m})_{k}^{*}(c_{m})_{k + 1}}\sum_{j = 0}^{k}(-1)^{j + k}|(c_{m})_{j}|^{2} ~,
        \end{aligned}
        \label{eqn:HmSol}
    \end{equation}
    where $k$ ranges from $1$ to $K$ (since $(c_{m})_{0} = 0$ for $m > 0$).    
\end{proof}

\section{Minimum required squeezing orders}\label{sec:minSqueezingOrders}
Our error-transparency constructions generally require the use of high orders of squeezing which are unlikely to be experimentally achievable in the near future.
Since these Hamiltonains are not necessarily the unique ones ET to the full pure-loss channel, we are left with the question of whether these squeezing orders are actually fundamentally required for exactly error-transparent amplitude-mixing gates for rotation-symmetric codes, or if they are simply a byproduct of our restricted parity nested structure and particular construction.
In this section, we make progress on this question by addressing the special case of gates which do not leave the codespace: we show that any unitary which is error-transparent to the full pure-loss channel, continuously maps the codespace to itself, and achieves an amplitude-mixing gate for a rotation-symmetric code must involve at least the squeezing orders used in the Theorem~\ref{thm:SqueezingScaling} Hamiltonian.
In particular, we have the following theorem:
\begin{theorem}
  Any Hamiltonian $H$ which is error-transparent (as according to Def.~\ref{def:ET}) to the error set $\mathcal{E}_{l}=\left\{\hat{a}^{m}\hat{n}^{k} \, \big| \, \frac{m}{2} + k \leq \frac{l}{2}\right\}$ and generates a continuous amplitude-mixing gate for an $\mathcal{E}_{l}$-correcting order-$N$ rotation-symmetric code must have nonzero matrix elements on at least $\lfloor l / 2 \rfloor + 1$ odd off-diagonals in the parity blocks of the Hamiltonian.
  In other words, it requires at least $\lfloor l / 2 \rfloor + 1$ squeezing orders to implement, specifically with orders of the form $Nk$ for $k$ odd.
  \label{thm:SqueezingRequired}
\end{theorem}

Here, we are implicitly dealing with time-independent Hamiltonians and thus single-pulse continuous amplitude-mixing gates.
This is sufficient since any general unitary which maps the codespace to itself at each time step and is not a phase gate must be built up from single-pulse continuous logical gates, at least one of which must must rotate around a non-$Z$ axis of the logical Bloch sphere (i.e. be amplitude-mixing).
Also, in order for this theorem to hold in broader contexts, we use a more general definition of error-transparency than the one presented in the main text~\cite{tsunoda_error-detectable_2023,xu_fault-tolerant_2024}:
\begin{definition}
    A Hamiltonian $H$ is \textbf{error-transparent} to an error set $\mathcal{E}$ if $[H, E] \in V_{\mathcal{E}}$ for any error $E \in V_{\mathcal{E}}$, where $V_{\mathcal{E}} = \text{span}\left(\left\{E'\mathcal{P}_{C} \, | \, E' \in \mathcal{E}\right\}\right)$ is the vector subspace of linear operators spanned by the errors restricted to the codespace ($\mathcal{P}_{C} = |0_{L}\rangle\langle 0_{L}| + |1_{L}\rangle\langle1_{L}|$ is the codespace projector).
    \label{def:ET}
\end{definition}
Note that, by the linearity of the Knill-Laflamme conditions and their restriction to the codespace, if $\mathcal{E}$ is correctable, then so is $V_{\mathcal{E}}$.
More specifically, the errors whose action on the codespace falls in $V_{\mathcal{E}}$ is precisely the set of errors that is corrected by any recovery operation which corrects $\mathcal{E}$.
The definition requires this set of errors to be closed under commutation with $H$, meaning the errors will commute `transparently' through the gate precisely in the sense that any idling recovery operation still corrects the errors regardless of when they occur during the gate.
This definition of error-transparency has the advantage of being linear in the Hamiltonian $H$ and the error set $\mathcal{E}$, and it generalizes the one presented in the main text and the one from~\cite{ma_error-transparent_2020}.
It is similar to the full error-closure conditions~\cite{tsunoda_error-detectable_2023}, which allow $V_{\mathcal{E}}$ to be replaced by any correctable error set that contains $V_{\mathcal{E}}$.
While satisfying the error-closure conditions still requires the Hamiltonian to map the errors to mutually correctable errors, the resulting errors are not necessarily correctable using the idling recovery operation, so the errors are no longer `transparent' in the sense of the main text that they are effectively the same as when idling.
In particular, weakening the definition to full error-closure may require large changes to recovery operations, and different recovery operations may be required following different gates.

Having noted the caveats and definitions of Theorem.~\ref{thm:SqueezingRequired}, we briefly remark on its implications before proceeding to the proof.
Essentially, this theorem reveals that all the squeezing orders in our constructions, or higher ones, are required if we want to achieve exactly error-transparent amplitude-mixing gates for rotation-symmetric codes.
Note that, unlike in Theorems.~\ref{thm:SqueezingScaling}-\ref{thm:SingleSqueezing}, this result holds for all exactly-correcting rotation-symmetric codes, not just ones with a finite number of fock-grid coefficients.
The only ways to possibly circumvent this theorem are to weaken to the error-closure conditions, which would require changes to the recovery operations, or allow the unitary to temporarily map the codespace elsewhere, where the errors are no longer guaranteed to be correctable, resulting in a much less straightforward concept of error-transparency.

\subsection{Reducing to Naive ET Construction}
The proof of Theorem.~\ref{thm:SqueezingRequired} hinges on a similar naive ET construction as in Eq.~\eqref{eqn:addedTerms}.
In particular, any continuous amplitude-mixing gate has logical state eigenvectors $|\psi_{L}\rangle = |0_{L}\rangle + \alpha|1_{L}\rangle$ and $|\psi_{L}^{\perp}\rangle = \alpha^{*}|0_{L}\rangle - |1_{L}\rangle$ for $\alpha \neq 0$, so we can write an ET Hamiltonian $H$ as
\begin{equation}
    \begin{aligned}
        H = \beta\sum_{E \in \mathcal{E}_{l}}\left(|\psi_{E}\rangle\langle\psi_{E}| - |\psi_{E}^{\perp}\rangle\langle\psi_{E}^{\perp}|\right) + \cdots ~,
    \end{aligned}
    \label{eqn:naiveET}
\end{equation}
where $\beta \neq 0$ is some overall scaling factor, $|\psi_{E}\rangle$ and $|\psi_{E}^{\perp}\rangle$ are the orthogonalized error words of $|\psi_{L}\rangle$ and $|\psi_{L}^{\perp}\rangle$, respectively, as in Eq.~\eqref{eqn:errorWords}, and $\cdots$ refers to further orthogonal terms of the eigendecomposition.
In other words, defining $V_{\psi} = \text{span}\left(\left\{E|\psi_{L}\rangle \, | \, E \in \mathcal{E}_{l}\right\}\right)$ to be the error space of $|\psi_{L}\rangle$, and $V_{\psi^{\perp}} = \text{span}\left(\left\{E|\psi_{L}^{\perp}\rangle \, | \, E \in \mathcal{E}_{l}\right\}\right)$ to be the error space of $|\psi_{L}^{\perp}\rangle$, we require $H$ to have $V_{\psi}$ in its $+1$ eigenspace and $V_{\psi^{\perp}}$ in its $-1$ eigenspace, where the Knill-Laflamme conditions guarantee that $V_{\psi}$ and $V_{\psi^{\perp}}$ are orthogonal.
We say this parity nested $H$ is naively error-transparent because it satisfies every known definition of error-transparency.
Beyond this, since adding multiples of $\mathcal{P}_{C}$ or $E\mathcal{P}_{C}$ to $H$ results in the same error-transparent gate just with different logical or error global phases, this construction even stipulates the precise global phase that is imparted on the codespace and each error space.
However, we find that it is sufficient to consider Hamiltonians of this naive form, meaning a restricted definition of error-transparency and our restricted parity nested structure do not effectively inhibit the ET Hamiltonians we can construct.
In particular, we split the proof of Theorem.~\ref{thm:SqueezingRequired} into two parts: in this section, we begin by proving that any Hamiltonian which generates a continuous amplitude-mixing gate and obeys the general ET conditions of Def.~\ref{def:ET} has the same matrix elements on the odd off-diagonals of each parity block as the naive Hamiltonian from Eq.~\eqref{eqn:naiveET}; in the following section, we then prove that any Hamiltonian having the eigenstructure of Eq.~\eqref{eqn:naiveET} must have nonzero matrix elements on at least $\lfloor l / 2 \rfloor + 1$ odd off-diagonals in the parity blocks of the Hamiltonian.

We now prove the following lemma:
\begin{lemma}
    Any Hamiltonian $H$ which generates a continuous amplitude-mixing gate and is error-transparent (Def.~\ref{def:ET}) to the full pure-loss channel $\mathcal{E}_{l}$ has the same matrix elements on the odd off-diagonals of each parity block as some Hamiltonian of the form of Eq.~\eqref{eqn:naiveET}.
    \label{lemma:reduceNaiveH}
\end{lemma}
\begin{proof}
    From Def.~\ref{def:ET}, we must have $[H, E\mathcal{P}_{C}] \in V_{\mathcal{E}_{l}} = \text{span}\left(\left\{E'\mathcal{P}_{C} \, | \, E' \in \mathcal{E}_{l}\right\}\right)$ for any $E \in V_{\mathcal{E}_{l}}$.
    Writing this out in terms of the logical eigenvectors $|\psi_{L}\rangle$ and $|\psi_{L}^{\perp}\rangle$ of the targeted continuous amplitude-mixing gate, we have:
    \begin{widetext}
        \begin{equation}
            \begin{aligned}
                HE\left(|\psi_{L}\rangle\langle\psi_{L}| + |\psi_{L}^{\perp}\rangle\langle\psi_{L}^{\perp}|\right) - E\left(|\psi_{L}\rangle\langle\psi_{L}| + |\psi_{L}^{\perp}\rangle\langle\psi_{L}^{\perp}|\right)H = \left(\sum_{E' \in \mathcal{E}_{l}}\alpha_{E'}E'\right)\left(|\psi_{L}\rangle\langle\psi_{L}| + |\psi_{L}^{\perp}\rangle\langle\psi_{L}^{\perp}|\right) ~,
            \end{aligned}
        \end{equation}
        for some linear combination coefficients $\alpha_{E'}$ for each $E' \in \mathcal{E}_{l}$.
        Abbreviating $A = \sum_{E' \in \mathcal{E}_{l}}\alpha_{E'}E'$ and letting $\lambda_{\psi}$ and $\lambda_{\psi^{\perp}}$ be the real eigenvalues of $|\psi_{L}\rangle$ and $|\psi_{L}^{\perp}\rangle$, respectively, with respect to $H$, we can write this as
        \begin{equation}
            \begin{aligned}
                \left(HE|\psi_{L}\rangle - \lambda_{\psi}E|\psi_{L}\rangle\right)\langle\psi_{L}| + \left(HE|\psi_{L}^{\perp}\rangle - \lambda_{\psi^{\perp}}E|\psi_{L}^{\perp}\rangle\right)\langle\psi_{L}^{\perp}| = \left(A|\psi_{L}\rangle\right)\langle\psi_{L}| + \left(A|\psi_{L}^{\perp}\rangle\right)\langle\psi_{L}^{\perp}| ~,
            \end{aligned}
            \label{eqn:errorSpaceClosureEqn}
        \end{equation}
    \end{widetext}
    meaning we must have $HE|\psi_{L}\rangle - \lambda_{\psi}E|\psi_{L}\rangle = A|\psi_{L}\rangle$ and $HE|\psi_{L}^{\perp}\rangle - \lambda_{\psi^{\perp}}E|\psi_{L}^{\perp}\rangle = A|\psi_{L}^{\perp}\rangle$.
    Since $\lambda_{\psi}E|\psi_{L}\rangle, A|\psi_{L}\rangle \in V_{\psi}$ and $\lambda_{\psi^{\perp}}E|\psi_{L}^{\perp}\rangle, A|\psi_{L}^{\perp}\rangle \in V_{\psi^{\perp}}$, we see from vector subspace closure that $HE|\psi_{L}\rangle \in V_{\psi}$ and $HE|\psi_{L}^{\perp}\rangle \in V_{\psi^{\perp}}$.
    Since this is true for any $E \in V_{\mathcal{E}_{l}}$, we see that $H$ must map both $V_{\psi}$ and $V_{\psi^{\perp}}$ to themselves.

    Therefore, there must exist some orthogonal basis of $V_{\psi}$ which are all eigenvectors of $H$.
    In particular, there must exist $\text{dim}(V_{\psi}) - 1$ errors $E_{i} \in V_{\mathcal{E}_{l}}$ such that $\mathcal{B} = \{|\psi_{L}\rangle\} \cup \{E_{i}|\psi_{L}\rangle \, | \, i \in \{1, \dots, \text{dim}(V_{\psi}) - 1\}\}$ are orthogonal eigenvectors of $H$ with eigenvalues $\lambda_{i}$.
    Re-deriving Eq.~\eqref{eqn:errorSpaceClosureEqn} with $E = E_{i}$ and using the fact that these $E_{i}|\psi_{L}\rangle$ are eigenvectors of $H$, we see that
    \begin{equation}
        \begin{aligned}
            \left(\lambda_{i} - \lambda_{\psi}\right)E_{i}|\psi_{L}\rangle\langle\psi_{L}| + \left(H - \lambda_{\psi^{\perp}}\right)E_{i}|\psi_{L}^{\perp}\rangle\langle\psi_{L}^{\perp}|
        \end{aligned}
        \label{eqn:errorSpaceClosureEigen}
    \end{equation}
    must be in $V_{\mathcal{E}_{l}}$.
    Consider the set of error \emph{operators} corresponding to the basis $\mathcal{B}$ of $V_{\psi}$: $\mathcal{B}_{\text{ops}} = \{I\} \cup \{E_{i} \, | \, i \in \{1, \dots, \text{dim}(V_{\psi}) - 1\}\}$.
    Since the elements of $\mathcal{B}$ are linearly independent and $V_{\psi}$ and $V_{\psi^{\perp}}$ are orthogonal, the elements of $\mathcal{B}_{\text{ops}}$ are also linearly independent.
    In addition, the number of elements in $\mathcal{B}_{\text{ops}}$ is equal to the dimension of $V_{\mathcal{E}_{l}}$ as an operator subspace, allowing us to conclude that $\mathcal{B}_{\text{ops}}$ is a basis for $V_{\mathcal{E}_{l}}$.
    Using this basis to construct the right-hand side of Eq.~\eqref{eqn:errorSpaceClosureEigen}, the first term of Eq.~\eqref{eqn:errorSpaceClosureEigen} dictates that the right-hand side must be equal to $(\lambda_{i} - \lambda_{\psi})E_{i}\mathcal{P}_{C}$.
    From this, we can conclude that $\left(H - \lambda_{\psi^{\perp}}\right)E_{i}|\psi_{L}^{\perp}\rangle = \left(\lambda_{i} - \lambda_{\psi}\right)E_{i}|\psi_{L}^{\perp}\rangle$, implying that $E_{i}|\psi_{L}^{\perp}\rangle$ is also an eigenvector of $H$, with eigenvalue $\lambda_{i^{\perp}} = \lambda_{i} - \lambda_{\psi} + \lambda_{\psi^{\perp}}$.
    Since the $E_{i}|\psi_{L}^{\perp}\rangle$ are also orthogonal (from the orthogonality of the $E_{i}|\psi_{L}\rangle$ and the Knill-Leflamme conditions), we have thus determined the eigenstructure of $H$ on the error spaces $V_{\psi}$ and $V_{\psi^{\perp}}$.
    In particular, the eigenstructure of $H$ must have the form
    \begin{equation}
        \begin{aligned}
            H = \sum_{i = 0}^{d}\left(\lambda_{i}|\psi_{i}\rangle\langle\psi_{i}| + \lambda_{i^{\perp}}|\psi_{i}^{\perp}\rangle\langle\psi_{i}^{\perp}|\right) + \cdots ~,
        \end{aligned}
    \end{equation}
    where $|\psi_{i}\rangle = E_{i}|\psi_{L}\rangle$, $|\psi_{i}^{\perp}\rangle = E_{i}|\psi_{L}^{\perp}\rangle$, $\lambda_{0} = \lambda_{\psi}$, $\lambda_{0^{\perp}} = \lambda_{\psi^{\perp}}$, $E_{0} = I$, and $d = \text{dim}\left(V_{\mathcal{E}_{l}}\right)$.
    The eigenvalue relation $\lambda_{i^{\perp}} = \lambda_{i} - \lambda_{\psi} + \lambda_{\psi^{\perp}}$ now stipulates that $\lambda_{0} - \lambda_{0^{\perp}} = \lambda_{1} - \lambda_{1^{\perp}} = \cdots = \lambda_{d} - \lambda_{d^{\perp}} = 2\beta$ for some constant $\beta \neq 0$.
    
    Writing the eigenvalues in the form $\lambda_{i} = \gamma_{i} + \beta$ and $\lambda_{i^{\perp}} = \gamma_{i} - \beta$ for constants $\gamma_{i}$, we can write $H$ as
    \begin{equation}
        \begin{aligned}
            H = & \beta\sum_{i = 0}^{d}\left(|\psi_{i}\rangle\langle\psi_{i}| - |\psi_{i}^{\perp}\rangle\langle\psi_{i}^{\perp}|\right) + \\
            & \sum_{i = 0}^{d}\gamma_{i}\left(|\psi_{i}\rangle\langle\psi_{i}| + |\psi_{i}^{\perp}\rangle\langle\psi_{i}^{\perp}|\right) + \cdots ~,
        \end{aligned}
        \label{eqn:gammaBetaEigenstructure}
    \end{equation}
    where, as before, $\cdots$ refers to terms of the eigendecomposition orthogonal to $V_{\psi}$ and $V_{\psi^{\perp}}$.
    We can recognize each of the terms in this second summation as proportional to $E_{i}\mathcal{P}_{C}E_{i}^{\dagger}$.
    Since the $E_{i}$ errors are correctable and arise from photon-loss plus dephasing, they can be written as some linear combination of the errors $\hat{a}^{m}\hat{n}^{k}\mathcal{P}_{C}$ with $m < N$.
    Since an error of this form only has nonzero matrix elements on the $m$-th off-diagonal, the $E_{i}$ can only have nonzero matrix elements on the first $N - 1$ off-diagonals.
    Thus, since the $|0_{L}\rangle$ and $|1_{L}\rangle$ codewords of rotation-symmetric codes have alternating support within the logical parity subspace, after projecting onto the $m$-th error parity subspace, the error words $E_{i}|0_{L}\rangle$ and $E_{i}|1_{L}\rangle$ will also have support on alternating Fock states within the $m$-th error parity subspace.
    Therefore, similar to the reasoning in Sec.~\ref{sec:derivingGates}, each $E_{i}|0_{L}\rangle\langle0_{L}|E_{i}^{\dagger} + E_{i}|1_{L}\rangle\langle1_{L}|E_{i}^{\dagger}$ pair of terms will only contribute nonzero matrix elements to the \emph{even} off-diagonals within each parity subspace.
    
    Thus, the second summation in Eq.~\eqref{eqn:gammaBetaEigenstructure} does not involve matrix elements on the odd off-diagonals in each parity block, so subtracting off this second summation from $H$ and dividing all the terms by the nonzero $\beta$, we see that $H$ has the same matrix elements on the odd off-diagonals of each parity block, up to the constant factor $\beta$, as the Hamiltonian $\sum_{i = 0}^{d}\left(|\psi_{i}\rangle\langle\psi_{i}| - |\psi_{i}^{\perp}\rangle\langle\psi_{i}^{\perp}|\right) + \cdots$, where the $\cdots$ terms are the same terms as in Eq.~\eqref{eqn:gammaBetaEigenstructure}.
    We can see that this summation eigenstructure has $V_{\psi}$ in its $+1$ eigenspace and $V_{\psi}^{\perp}$ in its $-1$ eigenspace, so is in fact exactly equivalent to a Hamiltonian of the form of Eq.~\eqref{eqn:naiveET}.
    This concludes the proof of Lemma.~\ref{lemma:reduceNaiveH}.
\end{proof}

\subsection{Naive ET required squeezing}
To finish proving Theorem.~\ref{thm:SqueezingRequired}, it remains to show that any naive ET Hamiltonian of the form of Eq.~\eqref{eqn:naiveET} must have nonzero matrix elements on at least $\lfloor l / 2 \rfloor + 1$ odd off-diagonals in the parity blocks.
Since by Lemma.~\ref{lemma:reduceNaiveH}, any general ET amplitude-mixing Hamiltonian is equal to a Hamiltonian of this form plus a contribution which does not change these odd off-diagonals, this statement then also holds for any general ET amplitude-mixing Hamiltonian.
We now complete the proof of Theorem.~\ref{thm:SqueezingRequired} in this way, which we have repeated below for convenience.
\newtheorem*{T3}{Theorem~\ref{thm:SqueezingRequired}}
\begin{T3}
  Any Hamiltonian $H$ which is error-transparent (as according to Def.~\ref{def:ET}) to the error set $\mathcal{E}_{l}=\left\{\hat{a}^{m}\hat{n}^{k} \, \big| \, \frac{m}{2} + k \leq \frac{l}{2}\right\}$ and generates a continuous amplitude-mixing gate for an $\mathcal{E}_{l}$-correcting order-$N$ rotation-symmetric code must have nonzero matrix elements on at least $\lfloor l / 2 \rfloor + 1$ odd off-diagonals in the parity blocks of the Hamiltonian.
  In other words, it requires at least $\lfloor l / 2 \rfloor + 1$ squeezing orders to implement, specifically with orders of the form $Nk$ for $k$ odd.
\end{T3}
\begin{proof}
    To show that at least $\lfloor l / 2 \rfloor + 1$ odd off-diagonals in the parity blocks are required, we consider a Hamiltonian $H$ with arbitrary matrix elements outside the parity blocks, on all the even off-diagonals in the parity blocks, and on $\lfloor l / 2 \rfloor$ of the odd off-diagonals in the parity blocks, with the remaining odd off-diagonals in the parity blocks having matrix elements fixed at zero.
    We then show that such a Hamiltonian cannot have the eigenstructure of Eq.~\eqref{eqn:naiveET} no matter what these arbitrary matrix elements are.
    This implies that more than $\lfloor l / 2 \rfloor$ odd off-diagonals in the parity blocks are required to achieve the eigenstructure of Eq.~\eqref{eqn:naiveET} and thus is sufficient to conclude the proof.
    To show this, we focus on the logical parity manifold (any parity subspace would work) and the associated $\hat{n}^{k}$ errors.
    The eigenstructure of Eq.~\eqref{eqn:naiveET} stipulates that $\hat{n}^{k}|\psi_{L}\rangle$ must be in the $+1$ eigenspace of $H$ and $\hat{n}^{k}|\psi_{L}^{\perp}\rangle$ must be in the $-1$ eigenspace of $H$ for all $k \in \{0, \dots, \lfloor l / 2 \rfloor\}$.
    In other words, using that these states only have support on the logical parity manifold, we must have $H_{0}\hat{n}^{k}|\psi_{L}\rangle = \hat{n}^{k}|\psi_{L}\rangle$ and $H_{0}\hat{n}^{k}|\psi_{L}^{\perp}\rangle = -\hat{n}^{k}|\psi_{L}^{\perp}\rangle$ for all $k \in \{0, \dots, \lfloor l / 2 \rfloor\}$.
    We will proceed by expressing these conditions as a linear system of equations in the matrix elements of $H_{0}$, and showing that there is no solution.
    In fact, it suffices to show there is no solution by just looking at the first row of $H_{0}$ (i.e. the $|0\rangle\langle Nk'|$ matrix elements for $k' \in \mathbb{N}$).

    Denote the $\lfloor l / 2 \rfloor$ odd off-diagonals which are allowed to be nonzero by the index set $J = \{j_{1}, j_{2}, \dots, j_{\lfloor l / 2 \rfloor}\}$ of $\lfloor l / 2 \rfloor$ distinct odd numbers.
    Recalling that $|\psi_{L}\rangle = |0_{L}\rangle + \alpha|1_{L}\rangle$ and $|\psi_{L}^{\perp}\rangle = \alpha^{*}|0_{L}\rangle - |1_{L}\rangle$ for $\alpha \neq 0$, we can then write this system of linear equations as
    \begin{widetext}
        \begin{equation}
            \begin{aligned}
                \langle 0|H_{0}\hat{n}^{k}|\psi_{L}\rangle & = \langle 0|\hat{n}^{k}|\psi_{L}\rangle \\
                \iff \sum_{j \in J}(H_{0})_{0, j}\left(\alpha j^{k}c_{j}\right) + \sum_{j \text{ even}}^{\infty}(H_{0})_{0, j}\left(j^{k}c_{j}\right) & = 0^{k}c_{0}
            \end{aligned}
            \label{eqn:linearSystem1}
        \end{equation}
        and
        \begin{equation}
            \begin{aligned}
                \langle 0|H_{0}\hat{n}^{k}|\psi_{L}^{\perp}\rangle & = -\langle 0|\hat{n}^{k}|\psi_{L}^{\perp}\rangle \\
                \iff \sum_{j \in J}(H_{0})_{0, j}\left(- j^{k}c_{j}\right) + \sum_{j \text{ even}}^{\infty}(H_{0})_{0, j}\left(\alpha^{*}j^{k}c_{j}\right) & = -0^{k}\left(\alpha^{*} c_{0}\right) ~,
            \end{aligned}
            \label{eqn:linearSystem2}
        \end{equation}
        for $k \in \{0, \dots, \lfloor l / 2 \rfloor\}$, where we have used the fact that $|0_{L}\rangle$ only has support on even Fock states and $|1_{L}\rangle$ only has support on odd Fock states in the logical parity manifold.
        This constitutes a system of $2\lfloor l / 2 \rfloor + 2$ equations in the infinite number of variables $(H_{0})_{0, j}$ for $j \in J$ or $j$ even (i.e. the potentially nonzero matrix elements in the first row of $H_{0}$).
        It suffices to show that this system of linear equations has no solution.
    
        We begin by rewriting the system of linear equations Eq.~\eqref{eqn:linearSystem1}-\eqref{eqn:linearSystem2} in matrix form as $M\vec{x} = \vec{b}$ :
        \begin{equation}
            \begin{aligned}
                \left(\begin{array}{ccc|cccc}
                    \alpha c_{j_{1}} & \cdots & \alpha c_{j_{\lfloor l / 2 \rfloor}} & c_{0} & c_{2} & c_{4} & \cdots \\
                    \alpha j_{1}c_{j_{1}} & \cdots & \alpha j_{\lfloor l / 2 \rfloor}c_{j_{\lfloor l / 2 \rfloor}} & 0 & 2c_{2} & 4c_{4} & \cdots \\
                    \alpha j_{1}^{2}c_{j_{1}} & \cdots & \alpha j_{\lfloor l / 2 \rfloor}^{2}c_{j_{\lfloor l / 2 \rfloor}} & 0 & 4c_{2} & 16c_{4} & \cdots \\
                    \vdots & \ddots & \vdots & \vdots & \vdots & \vdots & \ddots \\
                    \alpha j_{1}^{\lfloor l / 2 \rfloor}c_{j_{1}} & \cdots & \alpha j_{\lfloor l / 2 \rfloor}^{\lfloor l / 2 \rfloor}c_{j_{\lfloor l / 2 \rfloor}} & 0 & 2^{\lfloor l / 2 \rfloor}c_{2} & 4^{\lfloor l / 2 \rfloor}c_{4} & \cdots \\[6pt]
                    \hline
                    -c_{j_{1}} & \cdots & -c_{j_{\lfloor l / 2 \rfloor}} & \alpha^{*}c_{0} & \alpha^{*}c_{2} & \alpha^{*}c_{4} & \cdots \\
                    -j_{1}c_{j_{1}} & \cdots & -j_{\lfloor l / 2 \rfloor}c_{j_{\lfloor l / 2 \rfloor}} & 0 & 2\alpha^{*}c_{2} & 4\alpha^{*}c_{4} & \cdots \\
                    -j_{1}^{2}c_{j_{1}} & \cdots & -j_{\lfloor l / 2 \rfloor}^{2}c_{j_{\lfloor l / 2 \rfloor}} & 0 & 4\alpha^{*}c_{2} & 16\alpha^{*}c_{4} & \cdots \\
                    \vdots & \ddots & \vdots & \vdots & \vdots & \vdots & \ddots \\
                    -j_{1}^{\lfloor l / 2 \rfloor}c_{j_{1}} & \cdots & -j_{\lfloor l / 2 \rfloor}^{\lfloor l / 2 \rfloor}c_{j_{\lfloor l / 2 \rfloor}} & 0 & 2^{\lfloor l / 2 \rfloor}\alpha^{*}c_{2} & 4^{\lfloor l / 2 \rfloor}\alpha^{*}c_{4} & \cdots \\
                \end{array}\right)
                \left(\begin{array}{c}
                    (H_{0})_{0, j_{1}} \\ \vdots \\ (H_{0})_{0, j_{\lfloor l / 2 \rfloor}} \\[5pt]
                    \hline \\[-10pt]
                    (H_{0})_{0, 0} \\ (H_{0})_{0, 2} \\ (H_{0})_{0, 4} \\ \vdots
                \end{array}\right) & = \left(\begin{array}{c} c_{0} \\ 0 \\ 0 \\ \vdots \\ 0 \\\hline \\[-12pt] -\alpha^{*}c_{0} \\ 0 \\ 0 \\ \vdots \\ 0 \end{array}\right) ~,
            \end{aligned}
            \label{eqn:matrixEquation}
        \end{equation}
    \end{widetext}
    where for visual clarity we have added lines to separate the odd off-diagonal terms from the even ones and the $|\psi_{L}\rangle$ constraint equations from the $|\psi_{L}^{\perp}\rangle$ ones.
    This matrix equation has no solutions if $\vec{b}$ is not in the span of the columns of $M$, or equivalently, if $\text{rank}\left(\left(M \, \big| \, \vec{b}\right)\right) > \text{rank}\left(M\right)$.
    To show this is the case, we will make frequent reference to Vandermonde matrices, which are any $(n + 1) \times (m + 1)$ matrices of the form
    \begin{equation}
        \begin{aligned}
            V = \begin{pmatrix}
                1 & 1 & 1 & \cdots & 1 \\
                x_{0} & x_{1} & x_{2} & \cdots & x_{m} \\ x_{0}^{2} & x_{1}^{2} & x_{2}^{2} & \cdots & x_{m}^{2} \\ \vdots & \vdots & \vdots & \ddots & \vdots \\ x_{0}^{n} & x_{1}^{n} & x_{2}^{n} & \cdots & x_{m}^{n}
            \end{pmatrix} ~.
        \end{aligned}
    \end{equation}
    Such a matrix has determinant
    \begin{equation}
        \begin{aligned}
            \text{det}(V) = \prod\limits_{0 \leq i \leq j \leq m}(x_{j} - x_{i}) ~,
        \end{aligned}
    \end{equation}
    which is nonzero if and only if the $x_{i}$ are all distinct~\cite{macon_inverses_1958}.

    Looking first at the right submatrix of $M$, since the bottom-right submatrix is equal to the top-right submatrix multiplied by $\alpha^{*}$, any columns which span the top-right submatrix alone when restricted to the top rows must also span the entire right submatrix when including the whole column.
    Now, if we divide all the columns of this right submatrix by the corresponding coefficient $c_{i}$, we see that the top-right submatrix takes the form of an infinite Vandermonde matrix with distinct $x_{i} = 2i$ for all $i \in \mathbb{N}$, and powers up to $n = \lfloor l / 2 \rfloor$.
    This implies that any $\lfloor l / 2 \rfloor + 1$ columns of this submatrix will constitute a square Vandermonde matrix with distinct $x_{i}$ and thus will be non-singular.
    Thus, any $\lfloor l / 2 \rfloor + 1$ columns of the right submatrix are linearly independent and span the right submatrix.
    This allows us to effectively reduce $M$ to a finite-dimensional $(2\lfloor l / 2 \rfloor + 2) \times (2\lfloor l / 2 \rfloor + 1)$ matrix $M'$ without changing its column span by replacing the right submatrix with any $\lfloor l / 2 \rfloor + 1$ of its columns.

    In particular, we reduce to $M'$ by choosing the first column of the right submatrix, along with the arbitrary but distinct $(k_{1} / 2)\text{-th}, (k_{2} / 2)\text{-th}, \dots, (k_{\lfloor l / 2 \rfloor} / 2)\text{-th}$ columns.
    It then suffices to prove that the square matrix $A = \left(M' \, \big| \, \vec{b}\right)$ is full rank, since then it trivially has higher rank than $M'$.
    To show this, since both row and column operations preserve rank, we first divide the columns of the left submatrix of $A$ by the corresponding coefficients $\alpha j_{i}c_{j_{i}}$, then multiply the rows of the bottom submatrix by $-\alpha$, then divide the columns of the right submatrix (except the first) by the corresponding $k_{i}$, then subtract the top submatrix from the bottom one, and finally rearrange the rows and columns to arrive at the matrix
    \begin{widetext}
        \begin{equation}
            \begin{aligned}
                A' & = \left(\begin{array}{cc|ccc|ccc}
                    1 & 1 & 1 / j_{1} & \cdots & 1 / j_{\lfloor l / 2 \rfloor} & 1 / k_{1} & \cdots & 1 / k_{\lfloor l / 2 \rfloor} \\
                    -|\alpha|^{2} & |\alpha|^{2} & 1 / j_{1} & \cdots & 1 / j_{\lfloor l / 2 \rfloor} & -|\alpha|^{2} / k_{1} & \cdots & -|\alpha|^{2} / k_{\lfloor l / 2 \rfloor} \\[2pt]
                    \hline 
                    0 & 0 & 1 & \cdots & 1 & 1 & \cdots & 1 \\
                    0 & 0 & j_{1} & \cdots & j_{\lfloor l / 2 \rfloor} & k_{1} & \cdots & k_{\lfloor l / 2 \rfloor} \\
                    \vdots & \vdots & \vdots & \ddots & \vdots & \vdots & \ddots & \vdots \\ 
                    0 & 0 & j_{1}^{\lfloor l / 2 \rfloor} & \cdots & j_{\lfloor l / 2 \rfloor}^{\lfloor l / 2 \rfloor} & k_{1}^{\lfloor l / 2 \rfloor} & \cdots & k_{\lfloor l / 2 \rfloor}^{\lfloor l / 2 \rfloor} \\[4pt]
                    \hline
                    0 & 0 & 0 & \cdots & 0 & -|\alpha|^{2} - 1 & \cdots & -|\alpha|^{2} - 1 \\
                    0 & 0 & 0 & \cdots & 0 & \left(-|\alpha|^{2} - 1\right)k_{1} & \cdots & \left(-|\alpha|^{2} - 1\right)k_{\lfloor l / 2 \rfloor} \\
                    \vdots & \vdots & \vdots & \ddots & \vdots & \vdots & \ddots & \vdots \\ 
                    0 & 0 & 0 & \cdots & 0 & \left(-|\alpha|^{2} - 1\right)k_{1}^{\lfloor l / 2 \rfloor} & \cdots & \left(-|\alpha|^{2} - 1\right)k_{\lfloor l / 2 \rfloor}^{\lfloor l / 2 \rfloor} \\
                \end{array}\right) \\
                & = \left(\begin{array}{c|c|c}
                    X & \cdot & \cdot \\
                    \hline
                    0 & Y & \cdot \\
                    \hline
                    0 & 0 & Z
                \end{array}\right) ~.
            \end{aligned}
        \end{equation}
    \end{widetext}
    Rearranged in this convenient form, we can see that $\text{det}(A') = \text{det}(X)\text{det}(Y)\text{det}(Z)$.
    Here, $X$ is a $2 \times 2$ Vandermonde matrix with distinct $x_{i}$ given by the set $\{-|\alpha|^{2}, |\alpha|^{2}\}$; $Y$ is a $\lfloor l / 2 \rfloor \times \lfloor l / 2 \rfloor$ Vandermonde matrix with distinct $x_{i}$ given by the set $\{j_{1}, \dots, j_{\lfloor l / 2 \rfloor}\}$; and $Z$ is $-|\alpha|^{2} - 1 \neq 0$ times a $\lfloor l / 2 \rfloor \times \lfloor l / 2 \rfloor$ Vandermonde matrix with distinct $x_{i}$ given by the set $\{k_{1}, \dots, k_{\lfloor l / 2 \rfloor}\}$.
    Therefore, $X$, $Y$, and $Z$, all have nonzero determinants, and thus so does $A'$.
    This means $A'$, and thus also $A$, is full rank, allowing us to conclude that the system of linear equations given by Eq.~\eqref{eqn:matrixEquation} indeed has no solution.
    This concludes the proof that at least $\lfloor l / 2 \rfloor + 1$ odd off-diagonals in the parity blocks are required to achieve a naive ET Hamiltonian of the form of Eq.~\eqref{eqn:naiveET}, and thus also concludes the proof of Theorem.~\ref{thm:SqueezingRequired}.
\end{proof}

\section{Numerical Simulations}\label{app:performance}
To test the performance of a given error-transparent gate generated by Hamiltonian $H$, we compute the channel fidelity of the noisy $\bar{X}$ gate generated by $H$ (with recovery), relative to the perfect unitary evolution $U = \exp\left(-i\frac{\pi}{2}H\right)$.
In particular, similar to~\cite{albert_performance_2018}, we follow the evolution of the four codespace Pauli matrices
\begin{equation}
    \begin{aligned}
        P_{L} & = |0_{L}\rangle\langle 0_{L}| + |1_{L}\rangle\langle 1_{L}| \\
        X_{L} & = |0_{L}\rangle\langle 1_{L}| + |1_{L}\rangle\langle 0_{L}| \\
        Y_{L} & = i|1_{L}\rangle\langle 0_{L}| - i|0_{L}\rangle\langle 1_{L}| \\
        Z_{L} & = |0_{L}\rangle\langle 0_{L}| - |1_{L}\rangle\langle 1_{L}|
    \end{aligned}
\end{equation}
through the combined quantum channel $\mathcal{R}_{\kappa t} \circ \mathcal{L}_{\kappa}$, where $\mathcal{R}_{\kappa t}$ is the perfect recovery operation for the binomial code (described in~\cite{michael_new_2016}) and $\mathcal{L}_{\kappa}$ represents the Lindbladian evolution with Hamiltonian $H$, collapse operator $\hat{a}$, and error rate $\kappa$:
\begin{equation}
    \dot{\rho} = -\frac{i}{\hbar}[H, \rho] + \kappa\left(\hat{a}\rho\hat{a}^{\dagger} - \frac{1}{2}\left\{\hat{a}^{\dagger}\hat{a}, \rho\right\}\right) ~,
\end{equation}
integrated over the total gate time $t = \frac{\pi}{2}$.
We can then calculate the process fidelity of the gate as~\cite{hashim_practical_2025}
\begin{equation}
    F({\kappa}) = \frac{1}{8}\sum_{M \in \{P, X, Y, Z\}}\text{Tr}\left[M_{L}\mathcal{R}_{\kappa t}(\mathcal{L}_{\kappa}(U^{\dagger}M_{L}U))\right]
\end{equation}
for each error rate $\kappa$.

In Fig.~\ref{fig:performance} we calculated this gate fidelity for a range of $\kappa$ for each of our constructed Hamiltonians $H$ for the $N = K = 3$ binomial code.
We repeat this performance plot for the $N = K = 4$ and $N = K = 5$ binomial codes in Fig.~\ref{fig:performanceApp}.
As can be seen, the error scaling of the various gates indeed matches the specified orders of error-transparency of each of the Hamiltonians, as guaranteed by the Theorem~\ref{thm:SqueezingScaling} construction.
We also note that the single-squeezing improvement of the Theorem~\ref{thm:SingleSqueezing} construction increases as the binomial code becomes larger.

\begin{figure*}
    \centering
    \includegraphics[width = \textwidth]{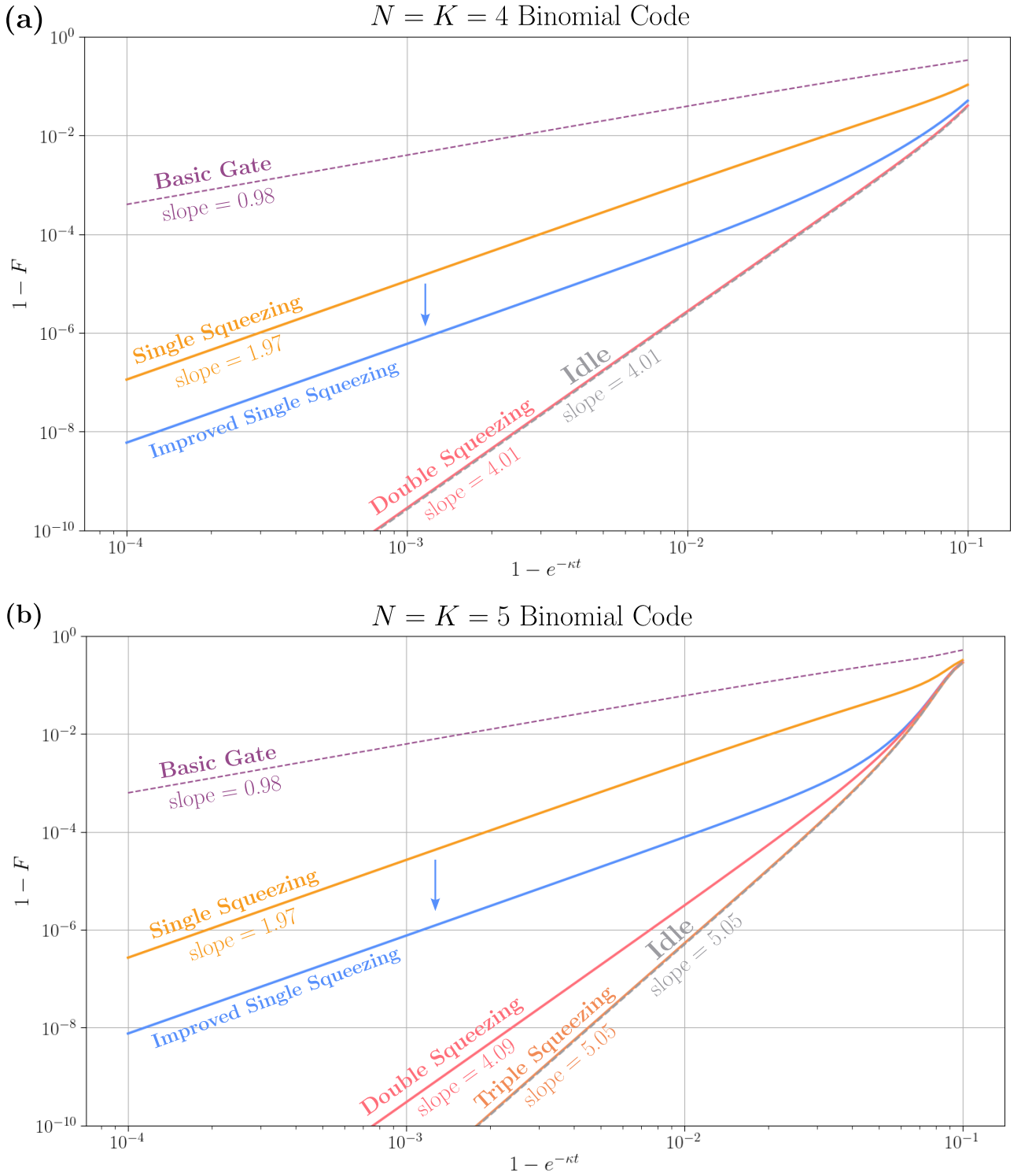}
    \caption{ 
        \textbf{Larger binomial code performance.}
        Repeat of performance plot (Fig.~\ref{fig:performance}) for the \textbf{(a)} $N = K = 4$ and \textbf{(b)} $N = K = 5$ binomial codes.
        The double squeezing (red) Hamiltonian comes from Theorem~\ref{thm:SqueezingScaling} with $l = 3$ (instead of $l = 2$), since these codes can accommodate three photon losses.
        The triple squeezing (dark orange) Hamiltonian comes from Theorem~\ref{thm:SqueezingScaling} with $l = 4$, and is not shown for the $K = N = 4$ binomial code since it cannot accommodate four photon losses.
        We note that the improvement (relative to the $l = 1$ Theorem~\ref{thm:SqueezingScaling} construction) of the improved single squeezing Theorem~\ref{thm:SingleSqueezing} construction increases as the binomial code becomes larger, yielding a factor of $19$ improvement for the $N = K = 4$ code and a factor of $35$ improvement for the $N = K = 5$ code. The slopes for the $N = K = 5$ binomial code fidelity curves are computed only based on the $1 - e^{-\kappa t} < 0.05$ points.
    }
    \label{fig:performanceApp} 
\end{figure*}

\section{Proof-of-Concept Implementation} ~\label{app:concept_implementation}
These error-transparent (ET) gates can be easily realized in different physical platforms for quantum information processing.
Here we show the possibility of their implementation on the smallest $N = K = 2$ binomial code using a bosonic storage mode coupled to an ancilla qubit.
Binomial code encoding and error correction has been performed in such an experimental setup~\cite{hu_quantum_2019,ni_beating_2023}.
For controlling the harmonic bosonic storage mode, some non-linearity is introduced by coupling a qubit to the storage mode. 
The qubit is used for state preparation, gate operations, error detection, error correction and probing the state in the storage mode.
The Hamiltonian for such a coupled system is given by:
\begin{equation}
    \begin{aligned}
        H_0 / \hbar =  \omega_c \hat{a}^\dagger \hat{a} + \omega_q \ket{e}\bra{e} + 
        \chi \hat{a}^\dagger \hat{a} \ket{e}\bra{e} ~.
    \end{aligned}
\end{equation}
Here $\omega_c$ is the oscillator frequency, $\omega_q$ is the qubit frequency, $\chi$ is the dispersive coupling, $\hat{a}$ is the annihilation operator for the storage mode and $|e\rangle$ is the excited state of the qubit. 

To achieve error-transparent gates on the $N=K=2$ binomial code we can use squeezing operations.
However, squeezing alone is not sufficient to accomplish unitary operations on the binomial encoding.
To accomplish a gate operation we need to modify the matrix elements of the driven oscillator which can be practically accomplished using the MEM protocol described in \cite{roy_synthetic_2025}.
In this protocol, we use a matrix element modification (MEM) frequency comb on the qubit with each frequency in the comb separated by $\chi$, the dispersive coupling of the qubit and the oscillator.
The number of frequencies in the frequency comb determines the Hilbert space size of the oscillator in which the dynamics is restricted, when starting from the ground state.
A double-frequency squeezing drive (of the form $\epsilon \hat{a}^2 + h.c.$) on the oscillator (at the frequency $2\omega_c, 2(\omega_c + \chi)$) induces dynamics in this limited Hilbert space of the oscillator.
Further, the phases ($\phi_n$) at each frequency of the qubit frequency comb is used to modify the matrix elements of the driven oscillator's annihilation operator. 
The expectation values of the $\hat{a}^2$ operator will change as $\langle n-2 | \hat{a}^2 | n\rangle = \sqrt{n(n-1)} \cos{ \left( (\phi_n-\phi_{n-2})/2 \right)}$ where $|n\rangle$ denotes the n-th Fock state of the oscillator.

For the $N=K=2$ binomial code, a frequency comb on the qubit with five frequencies (with spacing $\chi$) and appropriate phases ($\phi_n$) on the different frequency drives can be used to modify the matrix elements of the driven storage mode.
Then the storage drive can be applied for the required duration to accomplish a gate.
The drive Hamiltonian can be written as
\begin{widetext}
    \begin{equation}
        \begin{aligned}
            H_d={} \Omega \sum_{n=0}^{4} \cos\left(\omega_q t +  n \chi t + \phi_n + \frac{\pi}{2}\right)\hat{\sigma}_y -i \epsilon \sum_{k=0}^1 \cos\left( 2\omega_c t + 2k \chi t + \varphi+ \frac{\pi}{2}\right) (\hat{a}^2 - \hat{a}^{\dagger 2})
            \label{eqn:qubit_drive}
        \end{aligned}
    \end{equation}
\end{widetext}
where  $\Omega$ is the qubit Rabi drive rate, $t$ is the duration of the drive, $\phi_n$ is the phase of the drive at $n\chi$ shifted frequency, $\hat{\sigma}_y$ is the Pauli Y operator, $\epsilon$ is the drive rate of the storage mode, and $\varphi$ is the common drive phase for storage mode drive at both the frequencies. We set $\varphi=0$ for this work.
For the codespace, we can apply qubit drive at frequencies $\omega_q$, $\omega_q + 2\chi$ and $\omega_q + 4\chi$ and set the phases to $\phi_0 = \phi_2 = 0$ and $\phi_4 = 2.3$. 
For the error space evolution we have to apply qubit drive at frequencies $\omega_q + \chi$ and $\omega_q + 3\chi$ and set the phases to $\phi_1  = 0 $ and $\phi_3 = 2.0$.
Note that we need to change the phase on the $n=3$ qubit drive to match the codespace and error space evolution trajectory. 
Using these drives we can error-transparently create any superposition of $|0_L\rangle$ and $|1_L\rangle$ in the $XZ$ plane of the logical Bloch sphere for the $N = K = 2$ binomial code.
In Fig. \ref{fig:sf_bosonic_et}, we show the evolution of the states in the storage mode under such a drive when starting from $|1_L\rangle$. 
In this simulation, the system is lossless and in the adiabatic limit with $\epsilon:\Omega:|\chi| = 0.01: 1: 100$.
Fig.~\ref{fig:sf_bosonic_et} (a) and (b) shows the evolution in the code space and error space, respectively.
Driving the storage for the marked duration of 0.79$\epsilon^{-1}$, an error-transparent X gate can be accomplished.
At the end of the oscillator and qubit drives, we disentangle the oscillator and qubit by applying phases $-\phi_n/2$ on Fock level $|n\rangle$ of the oscillator. 
An additional drive duration dependent phase is added to Fock state $|4\rangle$ to stay in the codespace.
Both these phase additions can be done in the same step using a standard SNAP gate \cite{heeres_cavity_2015}.
At the end of the gate, we can check for photon loss errors by checking the parity of the state in the storage mode and then can correct the error.
This corrects for any photon loss error during the gate operation and hence achieves an error-transparent operation.

\begin{figure*}
    \centering
    \includegraphics[width = \textwidth]{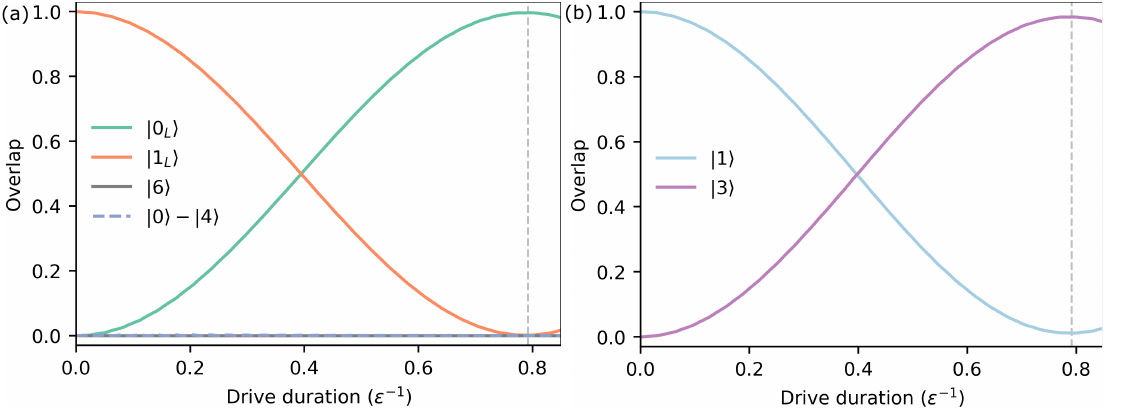}
    \caption{ 
        \textbf{Parity nested gate in superconducting cavity.} Implementation of the error-transparent gates on $N=K=2$ binomial code in a storage oscillator mode coupled to a qubit. Using an appropriate MEM comb on the qubit and driving the storage cavity for different duration we can accomplish error-transparent gates that can be used to create arbitrary superposition of the logical states with real coefficients. (a) and (b) respectively show the time evolution of the states in codespace and error space in the storage cavity. By driving the storage for the marked duration of 0.79$\epsilon^{-1}$, for both the error space and codespace, an error-transparent $\bar{X}$ gate can be accomplished. At the end of the gate we will need to detect errors by doing a parity check of the state in the storage cavity and then correct for any errors to accomplish the error-transparent gate. For these simulations the system was considered to be lossless and in the adiabatic limit with $\epsilon:\Omega:|\chi| = 0.01: 1: 100$.
    }
    \label{fig:sf_bosonic_et} 
\end{figure*}

While the above procedure shows the possibility of accomplishing such gates with superconducting circuits, experimental implementation will be challenging. 
The MEM protocol and simulation in Fig.~\ref{fig:sf_bosonic_et} assume adiabaticity (i.e. $\epsilon \ll \Omega \ll \chi$). 
To compete with decoherence in real experimental devices, we can not be in the strictly adiabatic regime. 
This will cause lower gate fidelities and less error transparent operations as the matrix element modification by the MEM protocol will not be perfect. 

To produce the squeezing drives shown in Eq.~\eqref{eqn:qubit_drive} we can use four wave mixing in a transmon qubit coupled to the oscillator~\cite{wang_efficient_2020}.
We can also use a superconducting nonlinear asymmetric inductive element (SNAIL) for creating faster squeezing and trisqueezing operations~\cite{eriksson_universal_2024}. 
However, these drives will also suffer from not being in the strictly adiabatic limit. 
For example, when we try to generate a high amplitude squeezing drive, a low amplitude linear displacement drive may also be generated. 
We can leverage numerical optimization techniques such as gradient descent or reinforcement learning to take into account these nonidealities and generate high fidelity error transparent operations starting from the MEM comb pulses as a seed.
We also note that the rate of squeezing achievable may also be a limiting factor in realization of ET gates in a short duration. 

In the main text, we only considered being error-transparent to photon loss in the oscillator. 
We are still susceptible to ancilla errors during the operation, in which case it will no longer be error-transparent. 
Nevertheless, we can still expect to see improvements in gate performance by being error-transparent to just oscillator photon loss as oscillator losses contribute significantly to
errors during gate operation.
In addition, in Appendix~\ref{app:ancillaET} we show how modifying the MEM protocol to use a four-level ancilla can ensure first-order error-transparency to ancilla errors.

For higher distance binomial codes, the required generalized order of squeezing can be achieved with higher order wave mixing processes.
A six-wave mixing process has recently been demonstrated in~\cite{vanselow_dissipating_2025}.
It might also be possible to design pulses numerically so that the operation is more robust to ancilla errors, but we leave this to be explored in future efforts.
In summary, while it is possible to achieve error-transparent amplitude-mixing operations with superconducting circuits, it will require significant progress in experimental and numerical techniques to outperform existing gate operations.

\section{Ancilla Error-Transparency}\label{app:ancillaET}
The MEM protocol~\cite{roy_synthetic_2025} can be made first-order error-transparent to ancilla errors by modifying the Hamiltonian to use a four-level ancilla.
This new Hamiltonian is given by
\begin{widetext}
    \begin{equation}
      \begin{aligned}
        \hat{H} / \hbar = & \hspace{4pt} \omega_{h}|h\rangle\langle h| + \omega_{e}|e\rangle\langle e| + \omega_{f}|f\rangle\langle f| + \omega_{c}\hat{n} + \hat{n}\left(\chi_{h}|h\rangle\langle h| + \chi_{e}|e\rangle\langle e| + \chi_{f}|f \rangle\langle f|\right) + \\
        & \hspace{4pt} \Omega_{gf}(t)\left(i|f\rangle\langle g| - i|g\rangle\langle f|\right) + \Omega_{he}(t)\left(i|e\rangle\langle h|  - i|h\rangle\langle e|\right) - i\epsilon(t)(\hat{a}^{2} - \hat{a}^{\dagger 2}) ~, \\
        \Omega_{gf}(t) = & \hspace{4pt} \bar{\Omega}\sum_{n = 0}^{N}\cos\left(\omega_{f}t + n\chi_{f}t + \phi_{n} + \frac{\pi}{2}\right) ~, \\
        \Omega_{he}(t) = & \hspace{4pt} \bar{\Omega}\sum_{n = 0}^{N}\cos\left(\left(\omega_{e} - \omega_{h}\right)t + n\left(\chi_{e} - \chi_{h}\right)t + \phi_{n} + \frac{\pi}{2}\right) ~, \\
        \epsilon(t) = & \hspace{4pt} \bar{\epsilon}\sum_{\chi = \{0, \chi_{h}, \chi_{f}\}}\cos\left(2\omega_{c}t + \chi t + \frac{\pi}{2}\right) ~,
      \end{aligned}
      \label{eqn:ancillaETHam}
    \end{equation}
    where $\hat{n} = \hat{a}^{\dagger}\hat{a}$ is the number operator of the oscillator with basis eigenkets $\{|n\rangle\,|\,n \in \mathbb{N}\}$ and $\{|g\rangle, |h\rangle, |e\rangle, |f\rangle\}$ are the basis kets of the four-level ancilla.
    Note that since the oscillator drive $\epsilon(t)$ does not include a $\chi_{e}$ term in the summation, we must enforce $\chi_{e} = \chi_{f}$ to properly achieve the MEM action.
    This is to avoid requiring $|\chi_{e} - \chi_{f}| \gg \bar{\Omega}$, which conflicts with the condition $\chi_{e} \approx \chi_{f}$ that is needed for error-transparency.
    
    To view the MEM action of this Hamiltonian, we successively switch into different interaction frames according to the following four unitaries:
    \begin{equation}
      \begin{aligned}
        \hat{U}_{1} & = \exp\left(i\left(\omega_{h}|h\rangle\langle h| + \omega_{e}|e\rangle\langle e| + \omega_{f}|f\rangle\langle f| + \omega_{c}\hat{n}\right)t\right) ~, \\
        \hat{U}_{2} & = \exp\left(i\left(\chi_{h}|h\rangle\langle h| + \chi_{e}|e\rangle\langle e| + \chi_{f}|f\rangle\langle f|\right)t\hat{n}\right) ~, \\
        \hat{U}_{3} & = \exp\left(\frac{i}{2}\sum_{n = 0}^{N}\phi_{n}|n\rangle\langle n| \otimes \left(|f\rangle\langle f| - |g\rangle\langle g| + |e\rangle\langle e| - |h\rangle\langle h|\right)\right) ~, \\
        \hat{U}_{4} & = \exp\left(\frac{i}{2}\bar{\Omega}t\sum_{n = 0}^{N}|n\rangle\langle n| \otimes \left(|f\rangle\langle g| + |g\rangle\langle f| + |e\rangle\langle h| + |h\rangle\langle e|\right)\right) ~.
      \end{aligned}
      \label{eqn:ancillaETUnitaries}
    \end{equation}
    Applying rotating-wave approximations between the frame switches according to the parameter hierarchy
    \begin{equation}
      \begin{aligned}
        \omega_{e} - \omega_{h}, \omega_{f}, \omega_{c} \gg |\chi_{f} - \chi_{h}|, |\chi_{h}|, |\chi_{e}|, |\chi_{f}| \gg \bar{\Omega} \gg \bar{\epsilon} ~,
      \end{aligned}
      \label{eqn:parameterHierarchy}
    \end{equation}
    the Hamiltonian ultimately becomes (see Appendix B of~\cite{roy_synthetic_2025} for a similar derivation)
    \begin{equation}
      \begin{aligned}
        \hat{H}' / \hbar & = \frac{\bar{\epsilon}}{2}\left(\sum_{n = 2}^{N}\sqrt{n(n - 1)}\cos\left(\frac{\phi_{n} - \phi_{n - 2}}{2}\right)|n - 2\rangle\langle n| + \sum_{N + 3}^{\infty}\sqrt{n}|n - 2\rangle\langle n|\right) + \text{h.c.} ~,
      \end{aligned}
      \label{eqn:MEMFinalHam}
    \end{equation}
\end{widetext}
so that the matrix-element modification is successfully accomplished for all states of the ancilla.

We now elaborate on the error-transparency of the Hamiltonian to first-order ancilla relaxation, using ideas from~\cite{reinhold_error-corrected_2020} and~\cite{xu_fault-tolerant_2024}.
Since this relaxation event induces ancilla dephasing, the situation for pure dephasing is similar.
The non-Hermitian effects of the errors similarly induce ancilla dephasing, so it suffices to only consider the error-transparency of the Hermitian evolution to the jump errors.
We consider an initial state of $|\psi_{\text{osc}, i}\rangle \otimes |g\rangle$, in the $\hat{U}_{\text{total}} = \hat{U}_{4}\hat{U}_{3}\hat{U}_{2}\hat{U}_{1}$ frame, where $|\psi_{\text{osc}, i}\rangle$ is some state of the oscillator that only has support up to the $N$-th Fock state.
Working in the $\hat{U}_{\text{total}}$ frame, we then follow the initial Hamiltonian evolution for time $t_{0}$, followed by a relaxation event at time $t_{0}$, followed by the remaining Hamiltonian evolution for time $T - t_{0}$, until the end of the gate is reached at time $T$.
For simplicity, we also choose the total gate time to be such that $\chi_{h}T = 2\pi k_{h}$, $\chi_{e}T = 2\pi k_{e}$, $\chi_{f}T = 2\pi k_{f}$, and $\bar{\Omega}T = 2\pi l$ for $k_{h}, k_{e}, k_{f}, l \in \mathbb{Z}$, so that at the beginning and end of the gate $\hat{U}_{\text{total}} = \hat{U}_{3}$.
This choice does not affect the error-transparency properties of the gate.

Since the ancilla stays in the $\{|g\rangle, |f\rangle\}$ subspace if there are no errors, the ancilla errors that matter to first order are the relaxation event $|e\rangle\langle f|$ and dephasing in the $\{|g\rangle, |f\rangle\}$ subspace (e.g. $|f\rangle\langle f|$).
In the frame associated with the total unitary $\hat{U}_{\text{total}} = \hat{U}_{4}\hat{U}_{3}\hat{U}_{2}\hat{U}_{1}$, the relaxation event becomes
\begin{widetext}
    \begin{equation}
      \begin{aligned}
        \hat{U}_{\text{total}}|e\rangle\langle f|\hat{U}_{\text{total}}^{\dagger} \propto & \hspace{4pt} \sum_{n = 0}^{N}\text{e}^{i(\chi_{e} - \chi_{f})nt}|n\rangle\langle n| \otimes \left|\left(\bar{\Omega}t, \pi / 2\right)_{eh}\right\rangle\left\langle\left(\bar{\Omega}t, \pi / 2\right)_{fg}\right| + \\
        & \sum_{n = N + 1}^{\infty}\text{e}^{i(\chi_{e} - \chi_{f})nt}|n\rangle\langle n| \otimes |e\rangle\langle f| ~,
      \end{aligned}
      \label{eqn:relexationError}
    \end{equation}
\end{widetext}
where $\left|\left(\theta, \phi\right)_{kj}\right\rangle = \cos(\theta / 2)|k\rangle + e^{i\phi}\sin(\theta / 2)|j\rangle$.
This error can be broken down into three parts: a time-dependent rotation of the oscillator $\text{e}^{(\chi_{e} - \chi_{f})\hat{n}t}$, the ancilla falling from the original qubit subspace $\{|g\rangle, |f\rangle\}$ to the error qubit subspace $\{|h\rangle, |e\rangle\}$, and dephasing of the ancilla via a collapse onto the state $\left|\left(\bar{\Omega}t, \pi / 2\right)_{eh}\right\rangle = \cos(\bar{\Omega}t/2)|e\rangle + i\sin(\bar{\Omega}t/2)|h\rangle$.
The pure dephasing event is similar, just without the random rotation and the fall into the $\{|h\rangle, |e\rangle\}$ subspace.
Similar to~\cite{reinhold_error-corrected_2020}, the random rotation error can be eliminated by fixing $\chi_{e} = \chi_{f}$, as already required by our Hamiltonian derivation.
Even if there is some small $\chi$ mismatch, if the Hamiltonian is error-transparent to random rotation errors -- or ET up to some order, like the parity nested Hamiltonians -- then we expect these dephasing errors can remain tolerable for small enough $(\chi_{e} - \chi_{f})t_{0}$, similar to~\cite{xu_fault-tolerant_2024}.

Following the evolution of the initial state through the channel and neglecting this random rotation error, we have:
\begin{widetext}
    \begin{equation}
      \begin{aligned}
        |\psi_{\text{final}}\rangle & \propto \text{e}^{i\hat{H}'(T - t_{0})}\left(\hat{U}_{\text{total}}|e\rangle\langle f|\hat{U}_{\text{total}}^{\dagger}\right)\bigg|_{t_{0}}\text{e}^{i\hat{H}'t_{0}}\left(|\psi_{\text{osc}, i}\rangle \otimes |g\rangle\right) \\
        & \propto \left(\text{e}^{i\hat{H}'T}|\psi_{\text{osc}, i}\rangle\right) \otimes \left(\cos\left(\bar{\Omega}t_{0} / 2\right)|e\rangle + i\sin\left(\bar{\Omega}t_{0} / 2\right)|h\rangle\right) ~.
      \end{aligned}
      \label{eqn:ETEvolution}
    \end{equation}
    Defining $|\psi_{\text{osc}, f}\rangle = \text{e}^{i\hat{H}'T}|\psi_{\text{osc}, i}\rangle$ and applying $\hat{U}_{3}^{\dagger}$ to return to the lab frame, we have:
    \begin{equation}
      \begin{aligned}
        \hat{U}_{3}^{\dagger}|\psi_{\text{final}}\rangle \propto \cos\left(\bar{\Omega}t_{0} / 2\right)\left(\hat{S}(-\vec{\phi} / 2)|\psi_{\text{osc}, f}\rangle \otimes |e\rangle\right) + i\sin\left(\bar{\Omega}t_{0} / 2\right)\left(\hat{S}(\vec{\phi} / 2)|\psi_{\text{osc}, f}\rangle \otimes |h\rangle\right) ~,
      \end{aligned}
      \label{eqn:labFinalState}
    \end{equation}
\end{widetext}
where $\hat{S}(\vec{\theta}) = \sum_{n}^{N}\text{e}^{i\theta_{n}}|n\rangle\langle n|$ is a SNAP operation.
Thus, if we measure the ancilla state in the lab frame to be $|e\rangle$, applying the SNAP operation $\hat{S}(\vec{\phi})$ -- which can be made ET~\cite{reinhold_error-corrected_2020} -- to the oscillator yields the correct final oscillator state $\hat{S}(\vec{\phi}/2)|\psi_{\text{osc}, f}\rangle = \text{Tr}_{A}\left[\hat{U}_{3}^{\dagger}\left(|\psi_{\text{osc}, f}\rangle \otimes |g\rangle\right)\right]$, where $\text{Tr}_{A}[\cdot]$ is the partial trace over the ancilla.
If we instead measure the state as $|h\rangle$, the oscillator is already in the correct final state $\hat{S}(\vec{\phi}/2)|\psi_{\text{osc}, f}\rangle$.
In the case of pure dephasing, if the measured ancilla state is $|f\rangle$, applying $\hat{S}(\vec{\phi})$ yields the correct final oscillator state $\hat{S}(\vec{\phi}/2)|\psi_{\text{osc}, f}\rangle$, and if the measured ancilla is $|g\rangle$, the oscillator is already in the correct final state.
This is consistent with the fact that the ancilla ends in $|g\rangle$ when there are no errors.
In this way, performing a measurement-conditioned SNAP operation at the end of the protocol ensures the Hamiltonian is first-order error-transparent to ancilla errors.
However, in an experiment, these operations will suffer from non-idealities (such as not being in the strictly adiabatic limit) similar to the MEM comb and will not be perfectly error transparent.

\end{document}